\documentclass{article}
\title{Generalized Load Balancing and Clustering Problems with Norm Minimization}
\author{Shichuan Deng\\
IIIS, Tsinghua University, China\\
\href{mailto:dsc15@mails.tsinghua.edu.cn}{\texttt{dsc15@mails.tsinghua.edu.cn}}}
\date{}

\usepackage{amsmath}
\usepackage{amsfonts,bbm}
\usepackage{amssymb}
\usepackage{graphicx,url,fullpage}
\usepackage[linesnumbered,ruled]{algorithm2e}
\usepackage{xspace}
\usepackage{amsthm}
\usepackage{enumerate}
\usepackage{hyperref}
\usepackage[capitalise]{cleveref}
\hypersetup{
    colorlinks=true,
    linkcolor=black,
    urlcolor=black,
    citecolor=blue
}

\newcommand{\calA}{\mathcal{A}}
\newcommand{\calH}{\mathcal{H}}
\newcommand{\calC}{\mathcal{C}}

\newcommand{\calF}{\mathcal{F}}
\newcommand{\calE}{\mathcal{E}}
\newcommand{\calD}{\mathcal{D}}
\newcommand{\calQ}{\mathcal{Q}}

\newcommand{\calW}{\mathcal{W}}

\newcommand{\calL}{\mathcal{L}}
\newcommand{\calU}{\mathcal{U}}
\newcommand{\calP}{\mathcal{P}}

\newcommand{\calI}{\mathcal{I}}
\newcommand{\calJ}{\mathcal{J}}
\newcommand{\calT}{\mathcal{T}}
\newcommand{\calM}{\mathcal{M}}

\newcommand{\up}[2]{{#1}^{(#2)}}

\newcommand{\R}{\mathbb{R}}
\newcommand{\rma}{r_0}

\newcommand{\ww}{\widetilde{w}}

\newcommand{\wt}{\mathsf{wt}}

\newcommand{\E}{\mathbb{E}}

\newcommand{\opt}{\mathsf{OPT}}

\newcommand{\etal}{\textit{et~al.}\xspace}

\newcommand{\dmax}{d_\mathsf{max}}

\newcommand{\spn}{\mathrm{span}}
\newcommand{\primal}[2]{\mathsf{Pl}_{#1}(#2)}
\newcommand{\dual}[2]{\mathsf{Dl}_{#1}(#2)}
\newcommand{\da}{{\downarrow}}

\newcommand{\queue}{\mathsf{queue}}

\newcommand{\almu}{{\alpha,\mu}}

\newcommand{\floor}[1]{\left\lfloor #1 \right\rfloor}
\newcommand{\ceil}[1]{\left\lceil #1 \right\rceil}
\newcommand{\topl}{\textrm{Top}_\ell}

\newcommand{\topq}[2]{\textrm{Top}_{#1}^{(#2)}}

\newcommand{\nextp}[1]{\mathsf{next}(#1)}
\newcommand{\prevp}[1]{\mathsf{prev}(#1)}
\newcommand{\counted}{\mathsf{ct}}
\newcommand{\que}{\mathsf{queue}}
\newcommand{\pos}{\mathsf{POS}}
\newcommand{\rbt}{R,B,T^\star_\ell}
\newcommand{\rbcalt}{R,B,\calT}

\newcommand{\pmax}{p_{\mathsf{max}}}
\newcommand{\pavq}{p_{\mathsf{av}}^{(q)}}

\newcommand{\algb}{\textsf{ALG-Bundle}\xspace}

\newcommand{\mnms}{\textsf{Max-Norm Makespan}\xspace}
\newcommand{\fmnms}{\textsf{Fair Max-Norm Makespan}\xspace}

\newcommand{\mnk}{\textsf{Max-Norm $k$Center}\xspace}
\newcommand{\mnknap}{\textsf{Max-Norm KnapCenter}\xspace}
\newcommand{\mnmat}{\textsf{Max-Norm MatCenter}\xspace}
\newcommand{\fmnk}{\textsf{Fair Max-Norm $k$Center}\xspace}

\newcommand{\mon}{maximum ordered norm\xspace}
\newcommand{\smn}{symmetric monotone norm\xspace}

\newtheorem{theorem}{Theorem}
\newtheorem{lemma}[theorem]{Lemma}
\newtheorem{corollary}[theorem]{Corollary}
\newtheorem{claim}[theorem]{Claim}
\newtheorem{observation}[theorem]{Observation}

\theoremstyle{definition}
\newtheorem{definition}[theorem]{Definition}
\theoremstyle{remark}
\newtheorem*{remark*}{Remark}
\newtheorem*{note*}{Note}

\begin{document}

\maketitle
\begin{abstract}
    In many fundamental combinatorial optimization problems, a feasible solution induces some real \emph{cost vectors} as an intermediate result, and the optimization objective is a certain function of the vectors. For example, in the problem of \emph{makespan minimization on unrelated parallel machines}, a feasible job assignment induces a vector containing the sizes of assigned jobs for each machine, and the goal is to minimize the $L_\infty$ norm of $L_1$ norms of the vectors. Another example is \emph{fault-tolerant $k$-center}, where each client is connected to multiple open facilities, thus having a vector of distances to these facilities, and the goal is to minimize the $L_\infty$ norm of $L_\infty$ norms of these vectors.

    In this paper, we study the \emph{maximum of norm} problem. Given an arbitrary \smn $f$, the objective is defined as the \emph{maximum} ($L_\infty$ norm) \emph{of $f$-norm} values of the induced cost vectors. This versatile formulation captures a wide variety of problems, including makespan minimization, fault-tolerant $k$-center and many others. We give concrete results for load balancing on unrelated parallel machines and clustering problems, including constant-factor approximation algorithms when $f$ belongs with a certain rich family of norms, and $O(\log n)$-approximations when $f$ is general and satisfies some mild assumptions.
    
    We also consider the aforementioned problems in a generalized fairness setting. As a concrete example, the insight is to prevent a scheduling algorithm from assigning \emph{too many} jobs consistently on any machine in a job-recurring scenario, and causing the machine's controller to fail. Our algorithm needs to stochastically output a feasible solution minimizing the objective function, and satisfy the given marginal fairness constraints.
\end{abstract}

\section{Introduction}

In many fundamental optimization problems, the objective is a function of several intermediate cost vectors induced by a feasible solution. One prominent example is load balancing on unrelated parallel machines, where a solution is an assignment of jobs to machines that induces a vector of job sizes for each machine. Another example is the $k$-clustering problem, where one needs to open at most $k$ facilities and connect each client to one or several facilities, resulting in a vector of assignment costs for clients. Depending on the formulation of the problem instance, various approaches of aggregating the vectors are used for evaluating the quality of the solution. In makespan minimization, the objective is the $L_\infty$ norm of $L_1$ norms of machine cost vectors. Many approximations are known~\cite{chakrabarty2015restricted,lenstra1990approximation,shmoys1993approximation,svensson2012santa}, with the best approximation ratio of 2 notoriously difficult to beat. This problem is also studied extensively for the more general ``$L_q$ of $L_1$'' cases, with the first 2-approximation given by Azar and Epstein in~\cite{azar2005convex}, and following improvements in~\cite{kumar2009unified,makarychev18solving}. Chakrabarty and Swamy~\cite{chakrabarty2019approximation,chakrabarty2019simpler} consider arbitrary {\smn}s of $L_1$ norms for load balancing and obtain constant approximations under mild assumptions on the input norm. Similarly, in clustering problems, $k$-center evaluates the solution by the $L_\infty$ norm of costs and $k$-median uses $L_1$ norm. These are two of the most studied clustering problems, and have a long line of research~\cite{arya2004local,byrka2017improved,charikar2002constant,charikar2012dependent,hochbaum1985best,hochbaum1986unified,jain2001approximation,li2013approximating}. Recently, ordered $k$-median has attracted a lot of interests as a generalization that unites $k$-center and $k$-median, and several constant approximations are developed~\cite{aouad2019ordered,byrka2018constant,chakrabarty2018interpolating}, with the current best ratio ($5+\epsilon)$ achieved by Chakrabarty and Swamy in~\cite{chakrabarty2019approximation}. In that paper, they also study $k$-clustering using general {\smn}s as the objective function, and obtain a constant approximation under mild assumptions of the input norm.

Motivated by the recent developments of more generalized functions in optimization problems (e.g.~\cite{chakrabarty2018generalized,chakrabarty2019approximation,ibrahimpur2020approximation}), we initiate the study of a natural and far-reaching general objective, capable of capturing many other well-studied problems. We consider the maximum of norm problem, where for a given \smn $f$, the goal is to find a feasible solution to the underlying combinatorial problem and minimize the maximum of $f$-norm values when evaluated on the induced cost vectors.\footnote{For simplicity, we always assume the input dimension of $f$ is larger and pad the input vector with 0s.} There are several cases of $f$ that are of particular interest, including $L_q$ norms where $q=\{1,2,\infty\}$, $\topl$ norms with $\ell\in\mathbb{Z}_+$ that sum up the $\ell$ largest entries in a vector, as well as the ordered norm $f(\vec{v})=w^T\vec{v}^\da$, where $w$ is a non-negative non-increasing vector and $\vec{v}^\da$ is the non-increasingly sorted version of $\vec{v}$. We note that $L_q$ norms and $\topl$ norms can be joint together under the same umbrella as $\topq{\ell}{q}$ norms, which is the $L_q$ norm of the $\ell$ entries in the vector that have the largest absolute values. The maximum of a finite number of ordered norms, known as a \emph{\mon}, is also of great significance in the family of {\smn}s. We refer to~\cite{chakrabarty2019approximation} and the references therein for a helpful discussion of such norms, as well as Chapter IV of the book~\cite{bhatia2013matrix} by Bhatia.

In our first concrete setting, we propose the following formal definitions for load balancing and clustering problems.\footnote{The knapsack and matroid variants of clustering are defined similarly except that we replace the cardinality constraint $|S|\leq k$ with a matroid constraint and a knapsack constraint. In a matroid constraint, we require $S\in\calI$ for $\calM=(\calF,\calI)$ being the given matroid; in a knapsack constraint, every $i\in\calF$ has a weight $\wt_i\geq0$ and we require $\wt(S)\leq W$.} Throughout this paper, we use $\calJ,[M]$ to denote the set of jobs and machines, respectively, also $\calC,\calF$ to denote the set of clients and candidate facility locations, with $d$ being the metric on $\calC\cup\calF$.

\begin{definition}(\mnms) Given jobs $\calJ$, machines $[M]$, and a \smn $f:\mathbb{R}_{\geq0}^{|\calJ|}\rightarrow\mathbb{R}_{\geq0}$, suppose the running time of job $j$ on machine $i$ is $p(i,j)\geq0$. The goal is to find a job assignment $\sigma$ that minimizes the objective
$\max_{i\in[M]}f(\vec{p}(i,\sigma))$, where $\vec{p}(i,\sigma):=\{p(i,j):\sigma(j)=i\}$.
\end{definition}

\begin{definition}(\mnk) Given a metric space $(\calC\cup\calF,d)$, integers $k,m\in\mathbb{Z}_+$, $l_j\leq r_j$ for $j\in\calC$ and a \smn $f:\mathbb{R}_{\geq0}^k\rightarrow\mathbb{R}_{\geq0}$, the goal is to identify a subset $S\subseteq \calF,|S|\leq k$, assign a subset of open facilities $S_j\subseteq S$ to every $j\in\calC$ satisfying $|S_j|\in[l_j,r_j],\sum_{j\in\calC}|S_j|\geq m$, and minimize the objective
$\max_{j\in\calC}f(\vec{d}(j,S_j))$, where $\vec{d}(j,S_j):=\{d(j,i):i\in S_j\}$.
\end{definition}

\begin{remark*}
We note that our definition of \mnk attains the utmost flexibility and generality. It captures the robust $k$-center problem by setting every $l_j=0,r_j=1$, as well as fault-tolerant $k$-center by letting $l_j=r_j,m=0$ and $f$ be the $L_\infty$ norm. One may also set $l_j=0,r_j=k$, rendering the individual constraints totally inert and letting the instance be controlled mainly by $k$ and $m$. We also argue that the current coverage constraint $m$ better suits the fault-tolerant and versatile nature of such multi-connection problems.
\end{remark*}

As a second result, we also consider the problems above in a fairness setting. In real world applications, the scenario of recurring scheduling and clustering tasks is quite plausible, and obtaining fairness in the long term is of great concern. For example, in load balancing problems, a deterministic algorithm may stress a single machine repeatedly with a large number of jobs and consequently reduce the reliability of certain parts, regardless of the actual makespan; in clustering problems, clients that are placed relatively poorly may always end up getting the least possible number of $l_j$ connections, which reduces the tolerance with errors and harms the user experience. We refer to~\cite{anegg2020technique,harris2019lottery} and the references therein for more motivating discussions.

Harris~\etal consider fair robust $k$-center in~\cite{harris2019lottery}, where a stochastic solution to robust $k$-center is required, which also ensures that the marginal probability of covering any $j$ is at least $p_j\in[0,1]$. They give a 2-pseudo-approximation, and also constant multi-criteria approximations for the matroid and knapsack variants. Anegg~\etal\cite{anegg2020technique} study the fair colorful $k$-center problem, where the clients are labeled with various colors representing different demographics. The stochastic clustering solution needs to satisfy a hard lower bound of covering $m_l$ clients with color $l$, as well as cover any client $j$ with probability at least $p_j\in[0,1]$. They give a true 4-approximation for this problem, improving the previous pseudo-approximation for fair robust $k$-center.

In our study, we generalize the notion of fairness. In \emph{fair maximum-of-norm makespan minimization} (abbreviated \fmnms), each machine is associated with a real number $e_i\geq0$, and we need to stochastically assign jobs so that the expected number of jobs on machine $i$ is at most $e_i$. In \emph{fair maximum-of-norm $k$-center} (abbreviated \fmnk), each client $j$ submits a real value $e_j\in[l_j,r_j]$, and we need to stochastically decide which facilities to open and how to establish connections for all clients, such that the expected number of connections for $j$ is at least $e_j$. We note that the coverage constraint $m$ is replaced by $\{e_j:j\in\calC\}$ in this problem, considering that both are quantitative constraints on the number of client-facility connections in the solution, and there is a technical difficulty in our approach when incorporating both.

\begin{definition}(\fmnms) Given jobs $\calJ$, machines $[M]$, real values $e_i\geq0$ for $i\in[M]$ and a \smn $f:\mathbb{R}_{\geq0}^{|\calJ|}\rightarrow\mathbb{R}_{\geq0}$, suppose the running time of job $j$ on machine $i$ is $p(i,j)\geq0$. The goal is to find the minimum value $\opt\geq0$, such that there exists a distribution $\calD$ on all possible assignments, satisfying $\forall i\in[M],\,\E_{\sigma\sim\calD}[|\sigma^{-1}(i)|]\leq e_i$ and $\Pr_{\sigma\sim\calD}[f(\vec{p}(i,\sigma))\leq\opt]=1$.
\end{definition}

\begin{definition}(\fmnk) Given a metric space $(\calC\cup\calF,d)$,  integers $k\in\mathbb{Z}_+$, $l_j\leq r_j$, real values $e_j\in[l_j,r_j]$ for $j\in\calC$ and a \smn $f:\mathbb{R}_{\geq0}^k\rightarrow\mathbb{R}_{\geq0}$, the goal is to find the minimum value $\opt\geq0$, such that there exists a distribution $\calD$ on subsets of $\calF$ and an efficient probabilistic algorithm for assigning $\{S_j\subseteq S:j\in\calC\}$ given $S\sim\calD$, such that: (1) $\forall j\in\calC$, $\Pr[|S_j|\in[l_j,r_j]]=1$ (coverage constraint); (2) $\forall j,\,\E[|S_j|]\geq e_j$ (fairness constraint); (3) $\forall j,\,\Pr[f(\vec{d}(j,S_j))\leq\opt]=1$ (cost constraint); (4) $\Pr[|S|\leq k]=1$ (cardinality constraint).
\end{definition}

%%%%%%%%%%%%%%%%%%%%%%%%%%%%%%%%%

\subsection{Our Results}

Our first algorithmic result contains approximation algorithms for \mnms and \mnk with varying approximation factors. We also adapt our algorithm to the more general clustering criteria with matroid and knapsack constraints, as can be found in~\cref{app:matcenter:ordered} and~\cref{app:knapcenter:ordered}.

\begin{theorem}\label{theorem:load:main}
    There exists a polynomial time algorithm for \mnms, such that
    \begin{enumerate}
        \item\label{theorem:load:main1} If $f$ is a \emph{$\topq{\ell}{q}$} norm, the approximation factor is $(4^{1/q}+\epsilon)$ in running time $poly(n)/\epsilon$;
        \item\label{theorem:load:main2} If $f$ is a \mon, the approximation factor is $O(\log n)$.
    \end{enumerate}
\end{theorem}

\begin{theorem}\label{theorem:center:main}
    There exists a polynomial time algorithm for \mnk, such that
    \begin{enumerate}
        \item\label{theorem:center:main1} If $f$ is a \emph{$\topq{\ell}{q}$} norm, the approximation factor is $(3\cdot 4^{1/q}+\epsilon)$ in running time $poly(n)/\epsilon$;
        \item\label{theorem:center:main2} If $f$ is a \mon, the approximation factor is $O(\log n)$.
    \end{enumerate}
\end{theorem}

In the second part, we extend our technique to the fairness setting, and show that the corresponding problems admit algorithms with similar approximation results.

\begin{theorem}\label{theorem:fair:load:topl}
	There exists a polynomial time algorithm for \fmnms that
    \begin{enumerate}
        \item\label{theorem:fair:load:topl1} If $f$ is a \emph{$\topq{\ell}{q}$} norm, the approximation factor is $(4^{1/q}+\epsilon)$ in running time $poly(n)/\epsilon$;
        \item\label{theorem:fair:load:topl2} If $f$ is a \mon, the approximation factor is $O(\log n)$.
    \end{enumerate}
\end{theorem}

\begin{theorem}\label{theorem:fair:center:topl}
    There exists a polynomial time algorithm for \fmnk that
    \begin{enumerate}
        \item\label{theorem:fair:center:topl1} If $f$ is a \emph{$\topq{\ell}{q}$} norm, the approximation factor is $(3\cdot 4^{1/q}+\epsilon)$ in running time $poly(n)/\epsilon$;
        \item\label{theorem:fair:center:topl2} If $f$ is a \mon, the approximation factor is $O(\log n)$.
    \end{enumerate}
\end{theorem}

\begin{remark*}
According to the characterizing Theorem 5.4 for {\smn}s by Chakrabarty and Swamy~\cite{chakrabarty2019approximation}, when the input norm $f$ (and the underlying problem) satisfies some mild assumptions, it can be well-approximated by a \mon. This provides a foothold for extending our $O(\log n)$-approximations to more general norms. We do not include the implied results here to avoid further complications, and refer to~\cite{chakrabarty2019approximation} for a more detailed discussion on the connection between {\mon}s and general {\smn}s.
\end{remark*}

\subsection{Outline of Techniques and Contributions}

We introduce several of our ingredients in proposing and solving the novel problems. The objective of a \smn is highly non-linear, even for the simple $\topl$ norm, therefore we adapt the sparsification frameworks that are used in solving ordered weighted optimization problems in~\cite{aouad2019ordered,byrka2018constant,chakrabarty2018interpolating,chakrabarty2019approximation}. However, unlike many previous successful algorithms in approximating the norm of a single vector, it seems that the current sparsification method is inadequate for approximating the optimum in our problem, where the objective is the maximum of norms of many vectors. This leaves an interesting $O(\log n)$-factor in the more general case, and perhaps reveals the intriguing nature of the problem.

Round-and-cut methods are widely used in many problem settings~\cite{anegg2020technique,chakrabarty2018generalized,li2016approximating,li2017uniform}.
For the fairness problems we study, we employ the primal-dual programs proposed in~\cite{anegg2020technique} for tackling the fairness constraint, which in turn uses the ellipsoid algorithm in the round-and-cut framework to find an approximate solution. To help the sparsification and round-and-cut methods work together, we exploit the structure of sparsification and construct a stronger separating hyperplane oracle for the ellipsoid algorithm.

In the clustering algorithms, we use a subroutine \algb (see~\cref{algorithm:bundle}) for partitioning the LP solution into so-called ``bundles''. Each \emph{full bundle} contains some fractional facilities that sum up to 1 total ``facility mass''. Moreover, we allow \emph{partial bundles} that have total facility mass less than 1. By defining a ``profitable factor'' $n_U\in\mathbb{Z}_+$ for a partial bundle $U$, if we choose not to open any facility in $U$, we do not have extra connections via $U$; if we choose to open one inside $U$, we gain $n_U$ additional connections. The auxiliary LP then carefully chooses some of the partial bundles and opens exactly one facility in each of them. A simpler version of the subroutine is used in~\cite{hajiaghayi2016constant,yan2015lp}, where only full bundles are allowed. This post-processing of the LP solution is also inspired by the Shmoys-Tardos algorithm for the generalized assignment problem~\cite{shmoys1993approximation}, where a machine copy may also be partial and contain ``job mass'' less than 1.

\subsection{Other Related Work}

In load balancing problems, much of the efforts are aimed at minimizing the norm of machine loads, while the load of each machine is fixed as the sum of job sizes that are assigned to it. We refer to the book by Williamson and Shmoys~\cite{williamson2011design} for a comprehensive overview. For the $L_\infty$ norm case, there are 2-approximations~\cite{lenstra1990approximation,shmoys1993approximation} that are notoriously hard to beat. Awerbuch~\etal\cite{awerbuch1995load} give a $\Theta(q)$-approximation for $L_q$ norms on unrelated machines. In a breakthrough result, Azar and Epstein~\cite{azar2005convex} use convex programming to show a 2-approximation for any fixed $L_q$ norm, which is subsequently improved by Kumar~\etal\cite{kumar2009unified} and Makarychev and Sviridenko~\cite{makarychev18solving}. In other ground-breaking results by Chakrabarty and Swamy~\cite{chakrabarty2019approximation,chakrabarty2019simpler}, they consider general {\smn}s of the load vector, and provide constant approximations when the norm $f$ admits a constant-factor ball-optimization oracle.

$k$-center is NP-hard to approximate for any factor smaller than 2~\cite{hsu1979easy}, and admits tight approximations~\cite{gonzalez1985clustering,hochbaum1985best}. As a closely related result, in the \emph{fault-tolerant $k$-center} problem, every $j\in X$ needs to be assigned at least $r_j$ distinct open facilities, with the cost defined as the longest connection it has, and tight constant approximations~\cite{chaudhuri1998p,khuller2000fault} are developed. In robust $k$-center, one only needs to cover $m$ clients. Charikar~\etal\cite{charikar2001algorithms} provide a 3-approximation for this problem, which is improved to a best-possible 2-approximation by Chakrabarty~\etal in~\cite{chakrabarty2020non} (see also Harris~\etal\cite{harris2019lottery}). Center clustering with knapsack constraints~\cite{hochbaum1986unified} and matroid constraints~\cite{chen2016matroid} are also studied, and Chakrabarty and Negahbani~\cite{chakrabarty2018generalized} give a unifying 3-approximation by considering a more general, down-closed family of subset constraints. Recently, Inamdar and Varadarajan~\cite{inamdar2019fault} study the robust fault-tolerant clustering problems for uniform $r_j$s, where they impose hard constraints on clients, and either discard a client $j$ completely or have to assign $r_j$ open facilities to it.

The celebrated $k$-median problem is famously known as APX-hard~\cite{jain2002greedy} and a great number of approximation algorithms are developed, e.g.~\cite{arya2004local,byrka2017improved,charikar2002constant,charikar2012dependent,jain1999primal,li2013approximating}. The fault-tolerant $k$-median problem is a variant where a client needs to be assigned at least $r_j$ open facilities, with its cost defined as the sum of distances to its assigned facilities. Swamy and Shmoys~\cite{swamy2008fault} develop a 4-approximation for the case when the clients have identical $r_j$ values, and Hajiaghayi~\etal\cite{hajiaghayi2016constant} give a constant approximation for non-uniform $r_j$s. Another prominent variant called ordered $k$-median also attracts a lot of research interests recently. In this formulation, the cost vector of clients is non-increasingly sorted and taken inner product with a predefined non-negative non-increasing weight vector. The first constant-factor approximation algorithms are given by Byrka~\etal\cite{byrka2018constant}, as well as Chakrabarty and Swamy~\cite{chakrabarty2018interpolating}. Chakrabarty and Swamy~\cite{chakrabarty2019approximation} later generalize ordered $k$-median to min-norm $k$-median, where the norm $f$ is symmetric and monotone, and devise a constant-factor approximation under some mild assumptions on $f$. They also improve the approximation ratio of ordered $k$-median to the currently best $(5+\epsilon)$.

\section{Preliminaries}

In this section, we provide the basic LP relaxations for load balancing and clustering problems. Throughout the paper, we use subscripts $l$ and $c$ to distinguish between the two problems. We use $i$ exclusively to refer to machines or facility locations, and $j$ for jobs or clients. 

In the basic LP for load balancing, denote $x_{ij}$ the extent of assignment from job $j$ to machine $i$ and fix $R\geq0$ as the largest job size allowed.
\begin{equation}
	\left\{
	\begin{array}{c}
		x\in[0,1]^{M\times|\calJ|}\\
	\end{array}
	\left|
	\begin{array}{cc}
		\sum_{i\in [M]}x_{ij}=1 & \forall j\in\calJ\\
		x_{ij}=0 & \forall (i,j)\in[M]\times\calJ\text{ s.t. }p(i,j) > R
	\end{array}
	\right.\right\}.\tag{$\calP_l(R)$}\label{lp:basic-load}
\end{equation}

For generic maximum-of-norm center, with no particular cardinality, knapsack or matroid constraints on facilities, $y_i$ denotes the extent of opening facility location $i\in\calF$, $u_j$ denotes the extent of connecting $j\in\calC$ to open facilities, and $x_{ij}$ denotes the fractional connection between $i\in\calF$ and $j\in\calC$. Fix some radius $R\geq0$ as the largest connection distance allowed.
\begin{equation}
	\left\{
	\begin{array}{c}
		x\in[0,1]^{|\calF|\times|\calC|}\\
		u\in\mathbb{R}^{|\calC|}\\
		y\in[0,1]^{|\calF|}
	\end{array}
	\left|
	\begin{array}{cc}
		\sum_{i\in\calF}x_{ij}=u_j & \forall j\in\calC\\
		l_j\leq u_j\leq r_j & \forall j\in\calC\\
		x_{ij}\leq y_i & \forall i\in\calF, j\in\calC\\
		x_{ij}=0 & \forall (i,j)\in\calF\times\calC\text{ s.t. }d(i,j) > R
	\end{array}
	\right.\right\}.\tag{$\calP_c(R)$}\label{lp:basic-cluster}
\end{equation}

\section{Maximum-of-Norm Load Balancing}

\subsection{Top$^{(q)}_\ell$ Norms}\label{section:load:topl}

We consider the case when $f$ is a Top$^{(q)}_\ell$ norm for $\ell\in\mathbb{Z}_+,q\geq1$, and refer to~\cref{app:load:ordered} for the case of {\mon}s. Suppose $\sigma^\star:\calJ\rightarrow [M]$ is the optimal assignment with optimum $\opt$, and for any machine $i\in[M]$, $j^\star(i,\ell)$ is the $\ell$-th largest job size among those $\sigma^{\star-1}(i)$ assigned for $i$. Evidently, one has $\opt\geq\left(\max_i\{\ell\cdot j^\star(i,\ell)^q\}\right)^{1/q}=\ell^{1/q}\cdot\max_i\{j^\star(i,\ell)\}$, and we define $T_\ell^\star:=\max_i\{j^\star(i,\ell)\}$. Also let $R=\max_{j}\{p(\sigma^\star(j),j)\}$ be the largest job size in the optimum, and obviously we have $R\leq\opt\leq|\calJ|^{1/q}\cdot R$. There are a polynomial number of possible values for $T^\star_\ell$ and $R$, and we can guess $B\in[\opt,(1+\epsilon)\opt)$ as integer powers of $(1+\epsilon)$ for fixed small $\epsilon>0$, since there are $O(\log_{1+\epsilon}|\calJ|)$ possible values. In what follows, we assume that we have guessed $R,T^\star_\ell$ and $B\in[\opt,(1+\epsilon)\opt)$ correctly.

Our algorithm begins with a straightforward LP relaxation,
\begin{alignat*}{3}
	\text{min\quad} && 0\tag{LB($\rbt$)}\label{lp:load-topl}\\
	\text{s.t.\quad}&&	x &\in\text{\ref{lp:basic-load}} &&\tag{LB($\rbt$).1}\label{lp:load-topl1}\\
&&	\sum_{j\in\calJ,p(i,j)\geq T_\ell^\star}x_{ij}&\leq \ell &&\forall i\in[M]\tag{LB($\rbt$).2}\label{lp:load-topl2}\\	
&&	\sum_{j\in\calJ,p(i,j)\geq T_\ell^\star}x_{ij}p(i,j)^q&\leq B^q &&\forall i\in[M].\tag{LB($\rbt$).3}\label{lp:load-topl3}
\end{alignat*}
\begin{claim}\label{claim:load-topl}
\emph{\ref{lp:load-topl}} is feasible.
\end{claim}
\begin{proof}
We transform the optimal assignment $\sigma^\star$ into an integral feasible solution $x^\star$, by letting $x^\star_{ij}=1$ if $\sigma^\star(j)=i$ and 0 otherwise. Obviously~\eqref{lp:load-topl1} is satisfied, since our guess $R$ is correct. For~\eqref{lp:load-topl2}, LHS is the number of jobs assigned to machine $i$ that have sizes at least $T^\star_\ell$, thus by definition is at most $\ell$. For~\eqref{lp:load-topl3}, LHS is the $L_q^q$ measure of jobs on $i$ that are at least $T^\star_\ell$ as large, and by definition of $T^\star_\ell$ again, is at most the Top$^{(q)}_\ell$ norm of jobs taken to the $q$-th power, thus at most $\opt^q\leq B^q$.
\end{proof}

\begin{proof}[Proof of~\cref{theorem:load:main},~\cref{theorem:load:main1}]
We briefly recall the procedures of Shmoys-Tardos algorithm in~\cite{shmoys1993approximation}. Given the fractional solution $x$, create $n_i=\ceil{\sum_{j\in\calJ}x_{ij}}$ copies of $i$. In non-decreasing order of $p(i,j)$, connect an extent of $x_{ij}$ from $j$ to the currently active copy of $i$, and overflow to the next copy if the current copy is full, i.e. has been assigned one unit of fractional jobs. We then use standard arguments for bipartite matching to round the fractional (sub)matching to a maximum integral (sub)matching, which also assigns every job to exactly one machine. Let $\hat\sigma(j)=i$ if $j$ is matched to some copy of $i$ and $\vec{p}(i,\hat\sigma)=\{p(i,j):\hat\sigma(j)=i\}$ be the vector of job sizes on $i$, and we have the following claim.
\begin{claim}\label{claim:load-topl-inline}
\emph{$\topq{\ell}{q}(\vec{p}(i,\hat\sigma))\leq \sqrt[q]{2R^q+B^q+\ell\cdot T^{\star q}_\ell}$}.
\end{claim}
\begin{proof}
We only consider the case when $\sum_{j\in\calJ}x_{ij}\geq\ell$, since the other one is easier.
Because the jobs are added in non-decreasing order during the construction of the fractional matching, the $\ell$ largest jobs are on the \emph{last $\ell$ full copies} of $i$, plus at most one job with size $\leq R$ on the possible \emph{partial copy} at the end.

Denote the $\ell$ full copies $U_{i1},\dots,U_{i\ell}$, in order of matching. Define $\pavq(i,U)$ the weighted average of $q$-th power job sizes on some copy $U$ and $\pmax(i,U)$ the maximum job size thereof. One has $\pmax(i,U_{it})^q\leq\pavq(i,U_{i(t+1)})$, because $U_{i(t+1)}$ has unit ``job mass'' assigned, and the fractional jobs are matched in non-decreasing order of $p(i,j)$, thus every job size in $U_{i(t+1)}$ is at least $\pmax(i,U_{it})$. The sum of $q$-th powers of these $\ell$ jobs are then at most
\begin{equation}
\sum_{t=1}^\ell \pmax(i,U_{it})^q+R^q\leq\sum_{t=1}^{\ell-1}\pavq(i,U_{i(t+1)})+2R^q\leq 2R^q+\sum_{j\in U_{it}:t=1,\dots,\ell}x_{ij}p(i,j)^q,\label{eq:load:next:inline}
\end{equation}
here, the union of $U_{it}:t=1,\dots,\ell$ contains exactly $\ell$ probability mass. It can also be broken down into two parts: those jobs with size $\geq T^\star_\ell$ and the rest. For the first part and using~\eqref{lp:load-topl3}, the sum is at most $B^q$. For the second part, since every job size is at most $T^\star_\ell$ and the total job mass is at most $\ell$, the sum is at most $\ell\cdot T^{\star q}_\ell$.

Therefore, $\topq{\ell}{q}(\vec{p}(i,\hat\sigma))\leq_{\text{\eqref{eq:load:next:inline}}} \sqrt[q]{2R^q+B^q+\ell\cdot T^{\star q}_\ell}$.
\end{proof}
For any $i$, using the claim and recall that $B\leq(1+\epsilon)\opt,\opt\geq\ell^{1/q}\cdot T^\star_\ell,\opt\geq R$, we conclude that $\hat\sigma$ is a $(4^{1/q}+\epsilon)$-approximation.
\end{proof}

\subsection{Fair Max-Norm Makespan}

In this section, we consider \fmnms for $\topq{\ell}{q}$ norms only. It is not hard to verify that the same algorithm also applies for general {\mon}s, so we omit the proof for~\cref{theorem:fair:load:topl2} of~\cref{theorem:fair:load:topl}. Recall our definitions of $\rbt$ in the previous section,
\begin{equation}
    B\geq R;\;B\geq\ell^{1/q}\cdot T^\star_\ell,\label{eq:fair:load:guess}
\end{equation}
and define $\calF_l(B)$ as the collection of all job assignments that correspond to the integral solutions of \ref{lp:load-topl} \emph{for all possible $R,T^\star_\ell$ that satisfy~\eqref{eq:fair:load:guess}}.
\[\calF_l(B):=\bigcup_{R,T^\star_\ell:\eqref{eq:fair:load:guess}}\{\sigma:x\in{\text{\ref{lp:load-topl}}}\cap\mathbb{Z}^{M\times|\calJ|},x_{ij}=1\Leftrightarrow\sigma(j)=i\},\]
as well as the truly feasible assignments $\calE_l(B):=\{\sigma:\topq{\ell}{q}(\vec{p}(i,\sigma))\leq B\forall i\in[M]\}$. It is easy to verify that $\calE_l(B)\subseteq\calF_l(B)$, and it makes more sense if we can directly work with $\calE_l(B)$. But $\calE_l(B)$ is more intractable using the current framework, so we relax the constraints and define the following LP and its dual using $\calF_l(B)$ instead, see~\cref{figure:load-primal-dual}.
\begin{figure}[ht]
\centering
\fbox{
\begin{minipage}[t]{0.47\linewidth}
\begin{flalign*}
	  \text{max\;\;}&0\tag{$\primal{l}{B}$}\label{lp:load-primal}\\
  \text{s.t.\;\;}&\forall i\sum_{\sigma\in\calF_l(B)}\lambda_\sigma\cdot\left|\sigma^{-1}(i)\right| \le e_i\\
 &\sum_{\sigma \in \calF_l(B)} \lambda_\sigma = 1,\;\lambda\ge 0.
  \end{flalign*}
\end{minipage}%
\;\vrule\;
\begin{minipage}[t]{0.47\linewidth}
  \begin{flalign*}
	  \text{min\;\;}&\sum_{i\in[M]}\alpha_i\cdot e_i-\mu\tag{$\dual{l}{B}$}\label{lp:load-dual}\\
  \text{s.t.\;\;}& \forall \sigma\in\calF_l(B)\sum_{i\in[M]}\alpha_i\cdot\left|\sigma^{-1}(i)\right|\ge\mu\\
  & \alpha \ge 0,\;\mu\in\mathbb{R}.
    \end{flalign*}
\end{minipage}%
}
    \caption{The primal and dual programs supported on $\calF_l(B)$.}
    \label{figure:load-primal-dual}
\end{figure}

Using a similar argument~\cite{anegg2020technique} and noticing that~\ref{lp:load-dual} is scale-invariant, if there exists a solution to~\ref{lp:load-dual} with a negative objective value, then it is unbounded and~\ref{lp:load-primal} is infeasible as a result. Therefore, we define the following polytope.
\begin{equation}
\left\{(\almu)\in\mathbb{R}_{\geq0}^M\times\mathbb{R}\left|
\begin{array}{cc}
	\sum_{i\in[M]}\alpha_i\cdot e_i\leq \mu-1 &\\
	\sum_{i\in[M]}\alpha_i\cdot|\sigma^{-1}(i)|\geq \mu & \forall \sigma\in\calF_l(B)
\end{array}
\right.\right\}.\tag{$\calQ_l(B)$}\label{lp:fair:load:q}
\end{equation}

\begin{claim}\label{claim:load:pvq}(\cite{anegg2020technique})
\ref{lp:fair:load:q} is empty if and only if \ref{lp:load-primal} is feasible.
\end{claim}

The core lemma is given as follows, which we prove in~\cref{app:fair:load:core-proof}. We are then ready to prove the main theorem.

\begin{lemma}\label{lemma:fair:load:core}
	Fix $B\geq0$. There is a polynomial time algorithm that, given $(\almu)\in\mathbb{Q}_{\geq0}^{M}\times\mathbb{Q}$ that satisfies $\sum_{i\in[M]}\alpha_i\cdot e_i\leq \mu-1$, either certifies that $(\almu)\in\calQ_l(B)$ or returns a feasible assignment $\sigma\in\calE_l(4^{1/q}B)$ with $\sum_{i\in[M]}\alpha_i\cdot|\sigma^{-1}(i)|< \mu$.
\end{lemma}

\begin{proof}[Proof of~\cref{theorem:fair:load:topl},~\cref{theorem:fair:load:topl1}]
Our algorithm iterates through all possible integer powers of $(1+\epsilon)$ as our guess $B\geq0$. Fix $B$ and we start with arbitrary $(\almu)$. If $\sum_i\alpha_i\cdot e_i>\mu-1$, then this constraint separates $(\almu)$ from $\calQ_l(4^{1/q}B)$ and we use the ellipsoid method to iterate for other values. Suppose $\sum_{i\in[M]}\alpha_i\cdot e_i\leq \mu-1$ is satisfied afterwards, and we call the algorithm in~\cref{lemma:fair:load:core} on $(\almu)$. If $(\almu)\in\calQ_l(B)$ is ever certified, then according to~\cref{claim:load:pvq}, \ref{lp:load-primal} is infeasible, and we abort from this choice of $B$. Otherwise, we always get some $\sigma\in\calE_l(4^{1/q}B)$ corresponding to a hyperplane which separates the current $(\almu)$ from $\calQ_l(4^{1/q}B)$ (since $\calE_l(4^{1/q}B)\subseteq\calF_l(4^{1/q}B)$). Using the ellipsoid method, we conclude the emptiness of $\calQ_l(4^{1/q}B)$ in polynomial time, and $\primal{l}{4^{1/q}B}$ is feasible by~\cref{claim:load:pvq} again.

When $B\in[\opt,(1+\epsilon)\opt)$, it is easy to verify that \ref{lp:load-primal} is feasible, hence \ref{lp:fair:load:q} is empty. In this case, our algorithm cannot find any $(\almu)\in\calQ_l(B)$ and has to verify the emptiness of $\calQ_l(4^{1/q}B)$. Moreover. let $\calH$ be the set of job assignments that are returned during the run of the ellipsoid algorithm and $\calH\subseteq\calE_l(4^{1/q}B)$ by definition. The following polytope is therefore infeasible
\begin{equation}
\left\{(\almu)\in\mathbb{R}_{\geq0}^M\times\mathbb{R}\left|
\begin{array}{cc}
	\sum_{i\in[M]}\alpha_i\cdot e_i\leq \mu-1 &\\
	\sum_{i\in[M]}\alpha_i\cdot|\sigma^{-1}(i)|\geq \mu & \forall \sigma\in\calH
\end{array}
\right.\right\}.\tag{$\calQ_l^\calH(4^{1/q}B)$}\label{lp:fair:load:qh4}
\end{equation}
By replacing the collection $\calF_l(B)$ with $\calH$ in \ref{lp:load-primal} and \ref{lp:load-dual}, we arrive at a conclusion similar to~\cref{claim:load:pvq}: $\mathsf{Dl}_l^\calH(4^{1/q}B)$ has optimum 0 and $\mathsf{Pl}_l^\calH(4^{1/q}B)$ is feasible. Since every assignment in $\calH$ is also in $\calE_l(4^{1/q}B)$ and $\calH$ has polynomial size, if we directly solve for the sampling probabilities $\{\lambda_\sigma:\sigma\in\calH\}$, the resulting solution satisfies the fairness constraint, and the maximum of norm is always bounded by $4^{1/q}B\leq4^{1/q}(1+\epsilon)\opt$.
\end{proof}

\section{Maximum-of-Norm $k$Center}

\subsection{Top$^{(q)}_\ell$ Norms}\label{section:center:topl}

In this section, we consider the case when $f$ is a $\topq{\ell}{q}$ norm and refer to~\cref{app:center:ordered} for general {\mon}s. Let $\rma=\max_{j\in\calC}\{r_j\}$. We first try guessing the longest client-facility connection distance $R$ in the optimal solution. We also know $\opt\in[R,\rma^{1/q} R]$ and can afford to guess $B\in[\opt,(1+\epsilon)\opt)$ for any small constant $\epsilon>0$. The number of such guesses is $O(\log_{1+\epsilon}\rma)$. 

Further, denote $(S^\star,\{S_j^\star\subseteq S^\star:j\in\calC\})$ the optimal solution, where $S^\star$ is the set of facility locations to open, and $S_j^\star$ is the set of open facilities that $j$ connects to. Specifically, we define $i(j,\ell,S^\star_j)$ the $\ell$-th largest connection distance for $j$ in the optimal solution, and define $T_\ell^\star:=\max_{j\in\calC}\{i(\ell,j,S^\star_j)\}$. Obviously, we can afford to guess $T_\ell^\star$ exactly as well. In what follows, we assume that we have the desired values $R,T_\ell^\star$ and $B\in[\opt,(1+\epsilon)\opt)$, and state the LP relaxation as follows,
\begin{alignat}{4}
	\text{min\quad} && 0\tag{Cl($\rbt$)}\label{lp:cluster-topl}\\
	\text{s.t.\quad}&&	(x,u,y)&\in\text{\ref{lp:basic-cluster}}&&\tag{Cl($\rbt$).1}\label{lp:cluster-topl1}\\
&&			\sum_{j\in\calC}u_j&\geq m&&\tag{Cl($\rbt$).2}\label{lp:cluster-topl2}\\
&&	\sum_{i\in\calF,d(i,j)\geq T_\ell^\star}x_{ij}&\leq \ell &&\forall j\in\calC\tag{Cl($\rbt$).3}\label{lp:cluster-topl3}\\	
&&	\sum_{i\in\calF,d(i,j)\geq T_\ell^\star}x_{ij}d(i,j)^q&\leq B^q &&\forall j\in\calC\tag{Cl($\rbt$).4}\label{lp:cluster-topl4}\\
&&	\sum_{i\in\calF}y_i&=k.&&\tag{Cl($\rbt$).5}\label{lp:cluster-topl5}
\end{alignat}
	
\begin{claim}\label{claim:cluster-topl}
    \emph{\ref{lp:cluster-topl}} is feasible.
\end{claim}
\begin{proof}
    The proof is essentially the same as~\cref{claim:load-topl}.
\end{proof}

For any solution $(x,u,y)$ to~\ref{lp:cluster-topl}, we assume that $x_{ij}$ is the best (nearest) assignment possible given $u,y$, and any facility location $i$ can be split into multiple co-located copies (see, e.g., ~\cite{charikar2012dependent,hajiaghayi2016constant}). Using the split technique, we further assume that $x_{ij}\in\{0,y_i\}$ and define $F_j=\{i\in\calF':x_{ij}>0\}$, with $\calF'$ being the set of facility locations after the split. Define $g:\calF'\rightarrow\calF$, which takes the copy to the original facility location. For example, if the original facility location $i\in\calF$ is split into 3 copies $i_1,i_2,i_3\in\calF'$, then $g(i_1)=g(i_2)=g(i_3)=i$, and $g^{-1}:\calF\rightarrow 2^{\calF'}$, $g^{-1}(i)=\{i_1,i_2,i_3\}\subseteq\calF'$ is the set of all copies of $i$.

\begin{algorithm}[ht]
\caption{\algb: an algorithm for creating bundles from a solution of~\ref{lp:basic-cluster}.}\label{algorithm:bundle}
\SetKwInOut{Input}{Input}
\SetKwInOut{Output}{Output}

\Input{$\calC,\calF',(x,u,y);\forall j\in\calC,F_j=\{i\in\calF':x_{ij}>0\}$}
\Output{$\calU_1,\calU_2,\{n_U:U\in\calU_2\};\forall j\in\calC,\queue_j=\{U_{j,l}:l=1,\dots,\ceil{u_j}\}$}
\tcc{split facilities into co-located copies if necessary}
$\calU_1\leftarrow\emptyset,\,\calU_2\leftarrow\emptyset;\forall j\in\calC,\,F_j'\leftarrow F_j,\,\queue_j\leftarrow\emptyset$\;
\While{exists $j\in\calC$ s.t. $|\queue_j|<\floor{u_j}$\label{algb-loop1}}{
Choose such $j$, so that if $U$ is the nearest unit volume of facilities in $F_j'$, $\max_{i\in U}d(i,j)$ is minimized\;
\eIf{exists $U'\in\calU_1, U\cap U'\neq\emptyset$}{
$\queue_j\leftarrow\queue_j\cup\{U'\}$, remove $U'\cap F_j'$ from $F_j'$\;
}{
$\queue_j\leftarrow\queue_j\cup\{U\}$, remove $U$ from $F_j'$, $\calU_1\leftarrow\calU_1\cup\{U\}$\;
}
}
\While{exists $j\in\calC$ s.t. $|\queue_j|<\ceil{u_j}$\label{algb-loop2}}{
Choose such $j$, so that if $U$ is the nearest \emph{at most unit} volume of facilities in $F_j'$, $y(U)$ is maximized\;
\uIf{exists $U'\in\calU_1, U\cap U'\neq\emptyset$}{
$\queue_j\leftarrow\queue_j\cup\{U'\}$, $F_j'\leftarrow\emptyset$\;
}
\uElseIf{exists $U'\in\calU_2, U\cap U'\neq\emptyset$}{
$\queue_j\leftarrow\queue_j\cup\{U'\}$, $F_j'\leftarrow\emptyset$, $n_{U'}\leftarrow n_{U'}+1$\;
}
\Else{$\queue_j\leftarrow\queue_j\cup\{U\}$, $F_j'\leftarrow\emptyset$, $\calU_2\leftarrow\calU_2\cup\{U\}$ and set $n_U=1$\;
}
}
\end{algorithm}

We use \algb (see~\cref{algorithm:bundle}) to process any solution $(x,u,y)$ to~\ref{lp:cluster-topl}, obtaining the output $\calU_1$ (full bundles), $\calU_2$ (partial bundles), $\{n_U:U\in\calU_2\}$ (profitable factors) and $\que_j$ of bundles for every $j\in\calC$. Define $\dmax(j,S)=\max_{i\in S}d(i,j)$, and we have the following lemma. The first part is from~\cite{hajiaghayi2016constant}, and the second part uses the fact that if $x_{ij}>0$, then $d(i,j)\leq R$, and it takes at most 3 hops from $j$ to any facility in $U_{j,t}$.
\begin{lemma}\label{lemma:3bundle} (\cite{hajiaghayi2016constant})
	Let $U_{j,t}$ be the $t$-th bundle added to $\queue_j$, then if $t\le\floor{u_j}$ and $V_{j,t}$ is the $t$-th closest unit mass of facilities in $F_j$ (split if necessary), we have $\dmax(j,U_{j,t})\leq3\dmax(j,V_{j,t})$; if $t=\floor{u_j}+1$, then $\dmax(j,U_{j,t})\leq 3R$.
\end{lemma}

Next we define the auxiliary LP for rounding.
\begin{alignat*}{3}
	\text{max} &&\quad\sum_{U\in\calU_2}n_U\cdot z(U)&+\sum_{j\in\calC}\sum_{U\in\queue_j}\mathbbm{1}\left[U\in\calU_1\right]\tag{ACl($\rbt$)}\label{lp:aux:k}\\
	\text{s.t.}&& z(U)&=1 \quad\quad\forall U\in\calU_1\tag{ACl($\rbt$).1}\label{lp:aux:k1}\\
	&& z(U)&\leq 1\quad\quad\forall U\in\calU_2\tag{ACl($\rbt$).2}\label{lp:aux:k2}\\
	&& z(g^{-1}(i'))&\leq1 \quad\quad\forall i'\in\calF\tag{ACl($\rbt$).3}\label{lp:aux:k3}\\
	&& z(\calF')&\leq k\tag{ACl($\rbt$).4}\label{lp:aux:k4}\\
	&& z_i&\geq0\quad\quad\forall i\in\calF'.\tag{ACl($\rbt$).5}\label{lp:aux:k5}
\end{alignat*}

\begin{lemma}\label{lemma:aux:k}
	\emph{\ref{lp:aux:k}} has integral extreme points with optimum at least $m$.
\end{lemma}

\begin{proof}
	The constraints of~\ref{lp:aux:k} contain two laminar families, $\calU_1\cup\calU_2$ and $\{g^{-1}(i'):i'\in\calF\}\cup\{\calF'\}$ (also see~\cite{hajiaghayi2016constant}), thus the extreme points are integral. We also know that $y$ is feasible to this LP, so we only need to prove the objective value of $y$ is at least $m$.
	
	Fix $j$ and consider the contribution of $\queue_j$. The first $\floor{u_j}$ bundles added to $\queue_j$ are all full bundles, so they all contribute to the second sum. If $u_j\neq\floor{u_j}$, then there is an additional bundle $U$ added. If $U$ is full, the contribution of $U$ in $\queue_j$ is also included in the second sum; otherwise, the contribution $y(U)$ of the partial bundle $U$ is counted in the first sum, because in \algb, whenever we add a partial bundle to $\queue_j$, we make sure to increase the counter $n_U$ by 1. Hence the objective value is exactly
	\[\sum_{j\in\calC}\sum_{U\in\queue_j}y(U)\geq\sum_{j\in\calC}\left(\floor{u_j}+(u_j-\floor{u_j})\right)=\sum_{j\in\calC}u_j\geq m,\]
	where the first inequality is because in the second loop of \algb, the partial bundles are sorted by their total mass in non-increasing order, hence any one added for $j$ has size at least $u_j-\floor{u_j}$. The last inequality is due to \eqref{lp:cluster-topl2}.
\end{proof}

We then simply compute the integral optimal solution $z^\star$ of~\ref{lp:aux:k}, and define the solution $\hat S=\{g(i):z_i^\star=1\}\subseteq\calF$. Let $\hat S_j=\{g(i):i\in U\in\que_j,z^\star_i=1\}$ be the set of open facilities $j$ connects to, which are found in its $\que_j$.

\begin{proof}[Proof of~\cref{theorem:center:main},~\cref{theorem:center:main1}]
First, the solution $\hat S$ and $\{\hat S_j:j\in\calC\}$ satisfy the coverage constraint $m$ according to~\cref{lemma:aux:k}. Fix any $j$, and we only consider the case when $u_j\notin\mathbb{Z}$ and $u_j>\ell$, since other cases are simpler. According to the rounding algorithm, there is exactly one open facility in each of $U_{j,1},\dots,U_{j,\floor{u_j}}$, and the worse case is when there exists another open facility in $U_{j,\ceil{u_j}}$. Using~\cref{lemma:3bundle}, the distances from $j$ to these open facilities are at most $3\dmax(j,V_{j,1}),\dots,3\dmax(j, V_{j,\floor{u_j}}),3R$, where we recall that $V_{j,t}$ is the $t$-th furthest unit mass of facilities in $F_j$. Thus the $q$-th power of $\topq{\ell}{q}(\vec{d}(j,\hat S_j))$ is at most
	\[\left(\sum_{t=\floor{u_j}-\ell+1}^{\floor{u_j}}(3\dmax(j, V_{j,t}))^q\right)+(3R)^q.\]
	The first $\ell-1$ terms in the sum above can each be bounded by the average distance $q$-th power from $j$ to the bundle following it, i.e., $\dmax(j,V_{j,t})^q\leq\sum_{i\in V_{j,t+1}}x_{ij}d(i,j)^q$, noticing that $V_{j,t+1}$ has unit mass and every facility in $V_{j,t+1}$ is at least $\dmax(j,V_{j,t})$ as far from $j$. Therefore, the sum is at most
	\begin{equation}
	(3R)^q+\sum_{t=\floor{u_j}-\ell+1}^{\floor{u_j}-1}\sum_{i\in V_{j,t+1}}x_{ij}(3d(i,j))^q+(3R)^q\leq2(3R)^q+3^q\sum_{t=\floor{u_j}-\ell+1}^{\floor{u_j}}\sum_{i\in V_{j,t}}x_{ij}d(i,j)^q.\label{eq:center-bundle-shift}
	\end{equation}
	Furthermore, the sum above is at most the sum of weighted distance $q$-th powers from $j$ to its furthest $\ell$ connected facility mass, denoted by $V_{\ell,\textrm{far}}$. According to \eqref{lp:cluster-topl3}, the minimum distance from $j$ to $V_{\ell,\textrm{far}}$ is \emph{at most} $T_\ell^\star$, hence the sum can be split into two parts,
	\begin{align}
\sum_{i\in V_{\ell,\textrm{far}}}x_{ij}d(i,j)^q&=\sum_{i\in V_{\ell,\textrm{far}},d(i,j)\geq T_\ell^\star}x_{ij}d(i,j)^q+\sum_{i\in V_{\ell,\textrm{far}},d(i,j)< T_\ell^\star}x_{ij}d(i,j)^q\notag\\
&\leq_{\text{\eqref{lp:cluster-topl4}}} B^q+T_\ell^{\star q}\sum_{i\in V_{\ell,\textrm{far}},d(i,j)< T_\ell^\star}x_{ij}\notag\\
&\leq B^q+\ell\cdot T_\ell^{\star q}.\label{eq:center-topl-split}
	\end{align}
	Combining \eqref{eq:center-bundle-shift}\eqref{eq:center-topl-split}, and the fact that $\opt\geq R,\opt\geq\ell^{1/q}\cdot T^\star_\ell$ and $(1+\epsilon)\opt\geq B$,
	\begin{equation}
	\topq{\ell}{q}(\vec{d}(j,\hat S_j))\leq \sqrt[q]{2(3R)^q+3^q(B^q+\ell\cdot T^{\star q}_\ell)}\leq3 (4^{1/q}+\epsilon)\opt.\label{eq:center:topl:inline}
	\end{equation}
	Finally, we take the maximum over $j$ and conclude the result.
\end{proof}

\subsection{Fair Max-Norm $k$Center}

We start by restricting the attention to distributions paired only with \emph{deterministic} assigning algorithms.
Consider $(\calD,\calA)$ as an optimal solution, where $\calD$ is a distribution on subsets of $\calF$, and $\calA$ is a probabilistic algorithm that outputs the assigned facilities $\{S_j\subseteq S:j\in\calC\}$ given the input $S\sim\calD$. We may always replace $\calA$ with $\calA'$, simply by greedily connecting $j$ to as many of its nearest open facilities as possible, as long as the norm of its cost vector is at most $\opt$ and there are at most $r_j$ connections. It is not hard to check that $(\calD,\calA')$ is also optimal. Therefore, we may only consider solutions that have such greedy deterministic algorithms. For $B\geq0$, we also define $\counted_B(j,S)$ as the maximum \emph{number} of facilities in $S$ that $j$ can connect to, as long as the connection cost vector has $f$-norm at most $B$ and the connection number is at most $r_j$.

Our round-and-cut framework is adapted from~\cite{anegg2020technique}. In general, we consider {\mon}s, but for conciseness, we only state the results for $\topq{\ell}{q}$ norms here and omit the proof for~\cref{theorem:fair:center:topl2} of~\cref{theorem:fair:center:topl}, as the algorithm essentially mirrors~\cref{app:center:ordered}. Recall our definitions for $\rbt$, and they must satisfy
\begin{equation}
    B\geq R;\;B\geq\ell^{1/q}\cdot T^\star_\ell.\label{eq:fair:center:guess}
\end{equation}

We also state our LP relaxations for the fair problem with no coverage or fairness constraints, differing from the orignal relaxation version~\ref{lp:cluster-topl}.
	\begin{alignat}{4}
	\text{min\quad}&& 0\tag{FCl$(\rbt)$}\label{lp:fair:center:topl}\\
	\text{s.t.\quad}&&(x,u,y)&\in\text{\ref{lp:basic-cluster}}\tag{FCl$(\rbt)$.1}\label{lp:fair:center:topl1}\\
	&&\sum_{d(i,j)\geq T^\star_\ell}x_{ij}&\leq \ell &\forall j\in\calC\tag{FCl$(\rbt)$.2}\label{lp:fair:center:topl2}\\
	&&\sum_{d(i,j)\geq T^\star_\ell}x_{ij}d(i,j)^q&\leq B^q &\forall j\in\calC\tag{FCl$(\rbt)$.3}\label{lp:fair:center:topl3}\\
	&&\sum_{i\in\calF}y_i&=k.\tag{FCl$(\rbt)$.4}\label{lp:fair:center:topl4}
	\end{alignat} 

Define $\calF_c(B)$ as the collection of all subsets of $\calF$ that correspond to the integral solutions of \ref{lp:fair:center:topl} \emph{for all possible $R,T^\star_\ell$ that satisfy~\eqref{eq:fair:center:guess}},
\[\calF_c(B):=\bigcup_{R,T^\star_\ell:\eqref{eq:fair:center:guess}}\{S\subseteq\calF:(x,u,y)\in{\text{\ref{lp:fair:center:topl}}}\cap\mathbb{Z}^{|\calF|\times|\calC|}\times\mathbb{Z}^{|\calC|}\times\mathbb{Z}^{|\calF|},y_i=1\Leftrightarrow i\in S\},\]
and the truly feasible subsets $\calE_c(B):=\{S\subseteq\calF:\counted_B(j,S)\in[l_j,r_j]\forall j\in\calC\}$. It is easy to verify that $\calE_c(B)\subseteq\calF_c(B)$. We define the following LP and its dual program based on $\calF_c(B)$,
\begin{figure}[ht]
\centering
\fbox{
\begin{minipage}[t]{0.47\linewidth}
\begin{flalign*}
	  \text{min\;\;}&0\tag{$\primal{c}{B}$}\label{lp:center-primal}\\
  \text{s.t.\;\;}&\forall j\sum_{S\in\calF_c(B)}\lambda_S\cdot\counted_B(j,S) \ge e_j\\
 &\sum_{S \in \calF_c(B)} \lambda_S = 1,\;\lambda\ge 0.
  \end{flalign*}
\end{minipage}%
\;\vrule
\begin{minipage}[t]{0.47\linewidth}
  \begin{flalign*}
	  \text{max\;\;}&\sum_{j\in\calC}\alpha_j\cdot e_j-\mu\tag{$\dual{c}{B}$}\label{lp:center-dual}\\
  \text{s.t.\;\;}& \forall S\in\calF_c(B)\sum_{j\in\calC}\alpha_j\cdot\counted_B(j,S)\le\mu\\
  & \alpha \ge 0,\;\mu\in\mathbb{R}.
    \end{flalign*}
\end{minipage}%
}
    \caption{The primal and dual programs supported on $\calF_c(B)$.}
    \label{figure:center-primal-dual}
\end{figure}

Clearly if~\ref{lp:center-primal} is feasible, its optimal value and the optimum of~\ref{lp:center-dual} are both 0. If there is some solution $(\almu)$ to~\ref{lp:center-dual} with a positive objective value, then since the constraints of~\ref{lp:center-dual} are scale-invariant,~\ref{lp:center-dual} is unbounded and thus~\ref{lp:center-primal} is infeasible. We define another polytope \ref{lp:fair:center:q} that is directed related to the feasibility of~\ref{lp:center-primal},
\begin{equation}
\left\{(\almu)\in\mathbb{R}_{\geq0}^{|\calC|}\times\mathbb{R}\left|
\begin{array}{cc}
	\sum_{j\in\calC}\alpha_j\cdot e_j\geq \mu+1 &\\
	\sum_{j\in\calC}\alpha_j\cdot\counted_B(j,S)\leq \mu & \forall S \in\calF_c(B)
\end{array}
\right.\right\}.\tag{$\calQ_c(B)$}\label{lp:fair:center:q}
\end{equation}

\begin{claim}\label{claim:center:pvq}(\cite{anegg2020technique})
	\ref{lp:fair:center:q} is empty if and only if~\ref{lp:center-primal} is feasible.
\end{claim}

The core lemma is given as follows, which we prove in~\cref{app:fair:center:core-proof}, and the proof of our main theorem on fairness immediately follows.

\begin{lemma}\label{lemma:fair:center:core}
	Fix $B\geq0$. There is a polynomial time algorithm that, given $(\almu)\in\mathbb{Q}_{\geq0}^{|\calC|}\times\mathbb{Q}$ that satisfies $\sum_{j\in\calC}\alpha_j\cdot e_j\geq \mu+1$, either certifies that $(\almu)\in\calQ_c(B)$ or returns a set $S\in\calE_c(3\cdot4^{1/q}B)$ with $\sum_{j\in\calC}\alpha_j\cdot\counted_{3\cdot4^{1/q}B}(j,S)> \mu$.
\end{lemma}

\begin{proof}[Proof of~\cref{theorem:fair:center:topl},~\cref{theorem:fair:center:topl1}]
Essentially the same as~\cref{theorem:fair:load:topl}, thus omitted here.
\end{proof}

%\input{ack}

%%
%% Bibliography
%%

\bibliography{cluster.bib}

\newpage
\appendix

\section{Maximum-of-Norm Load Balancing}

\subsection{MaxOrdered-Norm Makespan Minimization}\label{app:load:ordered}
In this section, we consider the case when $f$ is a \mon, i.e. $f(\vec{v})=\max_{w\in\calW}\{w^T\vec{v}^\da\}$, where $\calW=\{\up{w}{n}:n\in[N]\}$ is a finite set of non-increasing non-negative weight vectors. Likewise, let $\sigma^\star$ be the optimal assignment, $\opt$ be the optimum and $j^\star(i,\ell)$ be the $\ell$-th largest job size in jobs $\sigma^{\star-1}(i)$, and define $T^\star_\ell=\max_i\{j^\star(i,\ell)\}$. The number of such thresholds becomes $\omega(1)$ and each takes a polynomial number of possible values, thus guessing all of them precisely is prohibitive. Taking the idea of sparsification from~\cite{chakrabarty2019approximation}, define $\pos=\{\min\{2^t,|\calJ|\}:t\geq0\}$, $\nextp{\ell}=\min_{t\in\pos,t>\ell}\{t\}$ ($\nextp{|\calJ|}=|\calJ|+1$ and $\up{w}{n}_{|\calJ|+1}=0$), and simplified weight vectors,
\[\widetilde{\calW}=\{{\up{\ww}{n}}:n\in[N],\up{\ww}{n}_t=\up{w}{n}_t\text{ if }t\in\pos\text{ otherwise }\up{w}{n}_{\nextp{t}}\},\]
then for the \mon $\tilde{f}$ defined by $\widetilde{\calW}$, the optimum $\widetilde{\opt}$ satisfies (see~\cite{chakrabarty2019approximation})
\begin{equation}
    \tilde{f}(\vec{v})\leq f(\vec{v})\leq 2\tilde{f}(\vec{v}),\text{ especially }\widetilde{\opt}\leq\opt\leq2\widetilde{\opt}.\label{eq:load-approx}
\end{equation}

We first guess $R=T^\star_1$ and assume the value is correct (since it takes only $|\calJ|M$ possible values), then guess a \emph{non-increasing non-negative sequence} $\calT=\{T_\ell:\ell\in\pos\}$ such that it is supported on
$\{2^{-s}R:s\in\mathbb{Z}_{\geq0},2^{-s}R\geq R/|\calJ|\}\cup\{R/|\calJ|\}$, and
\begin{equation}
    T_\ell=\left\{\begin{array}{cc}
    \in[T^\star_\ell,2T^\star_\ell) & T^\star_\ell\geq R/|\calJ|\\
    = R/|\calJ| & \textrm{otherwise.}
    \end{array}\right.\label{eq:load-thresholds}
\end{equation}
It is easy to see that $|\calT|=O(\log|\calJ|)$, with each entry having $O(\log|\calJ|)$ choices, so routine calculations show that there are at most $2^{O(\log|\calJ|)}=poly(|\calJ|)$ possible sequences. 

\begin{lemma}\label{lemma:load-gap}
\begin{equation*}
    \max_{n\in[N]}\left\{\sum_{\ell\in\pos}(\up{\ww}{n}_\ell-\up{\ww}{n}_{\nextp{\ell}})\ell\cdot T_\ell
    \right\}=O(\log|\calJ|)\widetilde{\opt}.
\end{equation*}
\end{lemma}
\begin{proof}
Fix $n$ and $\ell$, define $\prevp{\ell}=\max_{t\in\pos:t<\ell}\{t\}$ ($\prevp{1}=0$). By definition of $T^\star_\ell$, the optimal solution at least satisfies
\begin{equation}
    \widetilde{\opt}\geq\left(\sum_{\ell'\leq\ell}\up{\ww}{n}_{\ell'}\right)T^\star_\ell=T^\star_\ell\cdot\sum_{\ell'\leq\ell}(\ell-\prevp{\ell})\up{\ww}{n}_{\ell'},\label{eq:opt-to-t}
\end{equation}
and if $T^\star_\ell\leq R/|\calJ|$, we have $T_\ell= R/|\calJ|$. Summing over $\ell\in\pos$, we obtain
\begin{align}
    &\sum_{\ell\in\pos}(\up{\ww}{n}_\ell-\up{\ww}{n}_{\nextp{\ell}})\ell\cdot T_\ell\leq_{\text{\eqref{eq:load-thresholds}}}\sum_{\ell\in\pos}(\up{\ww}{n}_\ell-\up{\ww}{n}_{\nextp{\ell}})\ell (R/|\calJ|+2T^\star_\ell)\notag\\
    &\leq_{\text{\eqref{eq:opt-to-t}}} R\up{\ww}{n}_1+2\sum_{\ell\in\pos} \frac{(\up{\ww}{n}_\ell-\up{\ww}{n}_{\nextp{\ell}})\ell}{\sum_{\ell'\leq\ell}(\ell-\prevp{\ell})\up{\ww}{n}_{\ell'}}\cdot\widetilde{\opt},\label{eq:load-log-items}
\end{align}
and notice that every fraction in the above is at most 1, because $\up{\ww}{n}$ is non-increasing. Since there are $O(\log|\calJ|)$ entries in $\pos$ and $\widetilde{\opt}\geq R\cdot\max_n\{\up{\ww}{n}_1\}$, \eqref{eq:load-log-items} is at most $O(\log|\calJ|)\widetilde{\opt}$, and the lemma is proven by taking the maximum over $n\in[N]$.
\end{proof}

\begin{remark*}
In fact,~\cref{lemma:load-gap} is asymptotically tight. Consider the instance with a single weight vector $w=(1,\sqrt{2}-1,\sqrt{3}-\sqrt{2},\dots,\sqrt{r}-\sqrt{r-1})$, where we let $r=2^t$ for simplicity, with the optimum achieved by many vectors simultaneously,
\begin{align*}
    \up{\vec{v}}{1} &= (1,0,0,\dots,0),\\
    \up{\vec{v}}{2} &= (1/\sqrt{2},1/\sqrt{2},0,\dots,0),\\
    \up{\vec{v}}{3} &= (1/\sqrt{3},1/\sqrt{3},1/\sqrt{3},\dots,0),\\
    &\cdots\\
    \up{\vec{v}}{r} &= (1/\sqrt{r},1/\sqrt{r},1/\sqrt{r},\dots,1/\sqrt{r}).
\end{align*}

Evidently, the optimum is $1$, and the LHS of~\cref{lemma:load-gap} is
\begin{align*}
    &(w_1-w_2)+(w_2-w_4)\cdot\sqrt{2}+(w_4-w_8)\cdot\sqrt{4}+\cdots+(w_{r/2}-w_r)\cdot\sqrt{r/2}+w_r\cdot\sqrt{r}\\
    =\,&w_1+w_2(\sqrt{2}-1)+w_4(\sqrt{4}-\sqrt{2})+\cdots+w_r(\sqrt{r}-\sqrt{r/2})\\
    \geq\,&\left(1-\frac{1}{\sqrt{2}}\right)\sum_{l=0}^t\left(\sqrt{2^l}-\sqrt{2^l-1}\right)\sqrt{2^l}\geq\left(\frac{1}{2}-\frac{1}{2\sqrt{2}}\right)t=\Omega(\log r).
\end{align*}

It is clear from the proof that, if $|\pos|$ is a constant, then LHS of~\cref{lemma:load-gap} is bounded by a constant times $\widetilde{\opt}$, too. However, in order for $\pos$ to induce a good approximation $\widetilde{\calW}$ and consequently a good approximation for $f$, a constant number of indexes in $\pos$ is almost never enough, unless the given \mon has a specific structure. Nevertheless, we do think the current sparsification framework still needs more exploring, especially in the presence of more complicated underlying problems.
\end{remark*}

We finally guess $B\in[\widetilde{\opt},2\widetilde{\opt})$, and notice by the triangle inequality of norm $\tilde f$, $\widetilde{\opt}\in[R\cdot\max_n\{\up{\ww}{n}_1\},|\calJ| R\cdot\max_n\{\up{\ww}{n}_1\})$, hence our number of guesses is polynomial-bounded. Suppose our guesses for the threshold values $\calT$ and $B$ are both as desired, our LP relaxation is the following.
\begin{align}
	\text{min} && \quad 0\tag{OLB($\rbcalt$)}\label{lp:load-ordered}\\
	\text{s.t.}&& x&\in\text{\ref{lp:basic-load}}\tag{OLB($\rbcalt$).1}\label{lp:load-ordered1}\\
	&&\sum_{j\in\calJ,p(i,j)\geq T_\ell}x_{ij}&\leq \ell\quad\;\,\forall i\in[M],\ell\in\pos\tag{OLB($\rbcalt$).2}\label{lp:load-ordered2}\\
	&&\sum_{\ell\in\pos}\left(\up{\ww}{n}_\ell-\up{\ww}{n}_{\nextp{\ell}}\right)
	\left(\sum_{p(i,j)\geq T_\ell}x_{ij}p(i,j)\right)&\leq B\quad\forall i\in[M],n\in[N].\tag{OLB($\rbcalt$).3}\label{lp:load-ordered3}
\end{align}

\begin{claim}\label{claim:load-ordered}
\emph{\ref{lp:load-ordered}} is feasible.
\end{claim}

\begin{proof}
Notice that $T_\ell\geq T^\star_\ell$ always holds true and the following conic expression of the ordered norm,
\begin{equation}
    \sum_{t}\up{\ww}{n}_t\cdot\vec{v}^\da_t=\sum_{\ell\in\pos}(\up{\ww}{n}_\ell-\up{\ww}{n}_{\nextp{\ell}})\topl(\vec{v}),\label{eq:conic-ordered}
\end{equation}
then an identical argument as in~\cref{claim:load-topl} shows the feasibility.
\end{proof}

\begin{proof}[Proof of~\cref{theorem:load:main},~\cref{theorem:load:main2}]
We apply Shmoys-Tardos algorithm again and let $\hat \sigma:\calJ\rightarrow[M]$ be the corresponding assignment and $\vec{p}(i,\hat\sigma)=\{p(i,j):\hat\sigma(j)=i\}$ be the vector of job sizes on machine $i$. Using an argument identical with~\cref{claim:load-topl-inline}, we obtain
\begin{equation}
    \topl(\vec{p}(i,\hat\sigma))\leq 2R + \sum_{j:p(i,j)\geq T_\ell}x_{ij}p(i,j) + \ell\cdot T_\ell,\label{eq:load-ordered-topl-inline}
\end{equation}
and for any $n\in[N]$, the true ordered cost is
\begin{align*}
    \sum_{t}\up{w}{n}_t\cdot\vec{p}(i,\hat\sigma)_t&\leq_{\text{\eqref{eq:load-approx}}}2\sum_{t}\up{\ww}{n}_t\cdot\vec{p}(i,\hat\sigma)_t=_{\text{\eqref{eq:conic-ordered}}}2\sum_{\ell\in\pos}(\up{\ww}{n}_\ell-\up{\ww}{n}_{\nextp{\ell}})\topl(\vec{p}(i,\hat\sigma))\\
    &\leq_{\text{\eqref{eq:load-ordered-topl-inline}}}2\sum_{\ell\in\pos}(\up{\ww}{n}_\ell-\up{\ww}{n}_{\nextp{\ell}})\left(2R + \sum_{j:p(i,j)\geq T_\ell}x_{ij}p(i,j) + \ell\cdot T_\ell\right)\\
    &\leq_{\text{\eqref{lp:load-ordered3}}} 4R\up{\ww}{n}_1+2B+2\sum_{\ell\in\pos}(\up{\ww}{n}_\ell-\up{\ww}{n}_{\nextp{\ell}})\ell\cdot T_\ell\\
    &=_{\text{\eqref{eq:load-approx}}}O(\log|\calJ|)\opt,
\end{align*}
where the last inequality is due to~\cref{lemma:load-gap}. Finally take the maximum over $i$ and $n$.
\end{proof}

\subsection{Proof of~\cref{lemma:fair:load:core}}\label{app:fair:load:core-proof}

First, we define the set of violating job assignments for notation simplicity,
\[\calF_l^{\almu}(B)=\left\{\sigma\in\calF_l(B)\left|\sum_{i\in[M]}\alpha_i\cdot|\sigma^{-1}(i)| <\mu\right.\right\},\]
and $\calE_l^\almu(B)$ similarly. We easily see that, if $(\almu)$ satisfies $\sum_{i\in[M]}\alpha_i\cdot e_i\leq\mu-1$ and $\calF_l^{\almu}(B)$ is empty, then $(\almu)\in\calQ_l(B)$. Therefore, to prove~\cref{lemma:fair:load:core}, for any given $(\almu)\in\mathbb{Q}_{\geq0}^M\times\mathbb{Q}$, we need to either certify that $\calF_l^{\almu}(B)$ is empty, or find a subset $S\in\calE_l^{\almu}(4^{1/q}B)$. The following lemma encodes the strict inequality as an equivalent non-strict one, and is directly obtained from~\cite{anegg2020technique}, thus we omit the proof here.

\begin{lemma}(\cite{anegg2020technique})
	Let $(\almu)\in\mathbb{Q}_{\geq0}^M\times\mathbb{Q}$ with representation length $L$. Then one can efficiently compute an $\eta>0$ with representation length $O(L)$, such that for any $\sigma\in\calF_l(B)$, we have 
	$\sum_{i\in[M]}\alpha_i\cdot|\sigma^{-1}(i)|<\mu$ if and only if $\sum_{i\in[M]}\alpha_i\cdot|\sigma^{-1}(i)|\leq\mu-\eta$.
\end{lemma}

We fix such an $\eta>0$ from now on, and define the modified polytopes
\[\calP_l^\almu(\rbt) = \left\{x\in\text{\ref{lp:load-topl}}\left|\sum_{i\in[M]}\alpha_i\cdot\sum_{j\in\calJ}x_{ij}\leq\mu-\eta\right.\right\}.\]

By definition, if $(\almu)$ satisfies $\sum_{i\in[M]}\alpha_i\cdot e_i\leq\mu-1$ but $(\almu)\notin\calQ_l(B)$, then there exists some assignment $\sigma\in\calF_l^\almu(B)$, hence $\sigma$ induces an integral solution to $\calP_l^\almu(\rbt)$ for \emph{some} $(\rbt)$ and $\calP_l^\almu(\rbt)$ is thus non-empty. The contrapositive of this observation tells us the following.
\begin{observation}\label{observation:fair:load}
Fix $B\geq0$. If $(\almu)$ satisfies $\sum_{i\in[M]}\alpha_i\cdot e_i\leq\mu-1$ and $\calP^\almu_l(\rbt)$ is empty \emph{for any} $(\rbt)$ that satisfies~\eqref{eq:fair:load:guess}, then $(\almu)\in\calQ_l(B)$.
\end{observation}

\begin{lemma}\label{lemma:load:separating}
	Let $(\almu)\in\mathbb{Q}_{\geq0}^M\times\mathbb{Q}$ and $(\rbt)$ satisfies~\eqref{eq:fair:load:guess}. There is a polynomial-time algorithm $\calA_0$ that, either returns $\sigma\in\calE_l^\almu(4^{1/q}B)$ or certifies that $\calP^\almu_l(\rbt)$ is empty.
\end{lemma}
\begin{proof}
The algorithm simply returns when $\calP_l^\almu(\rbt)$ is empty. If not, $\calA_0$ takes any feasible solution $x$ and attempts to find the desired output via rounding. This can be easily achieved via Shmoys-Tardos algorithm~\cite{shmoys1993approximation}, by duplicating machines, constructing the fractional (sub)matching and putting weight $\alpha_i$ on every edge that is incident on machine $i$ (and its copies). According to the definition of $\calP_l^\almu(\rbt)$, this fractional matching has weight at most $\mu-\eta$, and it is well-known that an integral sub-matching with at most the same weight can be found in polynomial time, corresponding to a job assignment $\sigma$. Since the weight is at most $\mu-\eta$, it means that $\sum_{i\in[M]}\alpha_i\cdot|\sigma^{-1}(i)|\leq\mu-\eta$ is satisfied and this is a violating constraint. Using an identical argument with~\cref{claim:load-topl-inline}, the maximum of $\topq{\ell}{q}$ norm under $\sigma$ is bounded by $\sqrt[q]{2R^q+B^q+\ell\cdot T^{\star q}_\ell}\leq_{\text{\eqref{eq:fair:load:guess}}}4^{1/q}B$, therefore $\sigma\in\calE_l^\almu(4^{1/q}B)$ and $\calA_0$ returns it.
\end{proof}

\begin{proof}[Proof of~\cref{lemma:fair:load:core}] We use the algorithm $\calA_0$ in~\cref{lemma:load:separating} as part of the separation hyperplane oracle, in order to check whether $\calP^\almu_l(\rbt)$ is empty for any feasible $(\rbt)$ . Fix $B,(\almu)$ and iterate through all possible $(\rbt)$ that satisfies~\eqref{eq:fair:load:guess}, and start with any $x$ in each iteration: 
\begin{itemize}
    \item Whenever the algorithm $\calA_0$ is called and returns some $\sigma\in\calE_l^\almu(4^{1/q}B)$, then $(\almu)$ violates $\calQ_l(4^{1/q}B)$ and we simply return $\sigma$;
    \item Assume $\calA_0$ verifies the emptiness of $\calP^\almu_l(\rbt)$ for all $(\rbt)$, in polynomial time. Using~\cref{observation:fair:load}, one has $(\almu)\in\calQ_l(B)$ verified.\qedhere
\end{itemize} 
\end{proof}

%%%%----%%%%----

%%--%%

\section{Maximum-of-Norm Center Clustering}

\subsection{MaxOrdered-Norm $k$Center}\label{app:center:ordered}

In this section, we consider the case where $f$ is defined by $f(\vec{v})=\max_{w\in \calW}\{w^T\vec{v}^\da\}$ and $\calW=\{\up{w}{n}:n\in[N]\}$ is a finite set of non-negative non-increasing weight vectors. We first guess $R$ exactly, which is the largest connection in the optimal solution. It is easy to see that $\opt\geq R\cdot\max_{n\in[N]}\{\up{\ww}{n}_1\}$. We then try to guess $T_\ell^\star$, which one recalls is the maximum of $\ell$-th largest connection distance for each client. To this end, we reuse the idea and notation from~\cref{app:load:ordered}, and guess $\calT=\{T_\ell:\ell\in\pos\}$ such that 
\[T_\ell\left\{
\begin{array}{cc}
	\in[T^\star_\ell,2T^\star_\ell) & T^\star_\ell\geq  R/\rma\\
	=  R/\rma & \text{otherwise,}
\end{array}
\right.\]
and the sparsified weight vectors $\widetilde{\calW}$. It is easy to see that, the number of sequences $\calT$ is at most $2^{O(\log\rma)}=poly(\rma)$, and we assume we now have the sequence we desired as well. Finally, it takes another logarithmic number of guesses to guess $B\in[\widetilde{\opt},2\widetilde{\opt})$.

Our LP relaxation is the following,
	\begin{align}
	\text{min} && \quad 0\tag{OCl($\rbcalt$)}\label{lp:center:ordered}\\
	\text{s.t.}&&(x,u,y)&\in\text{\eqref{lp:basic-cluster}}&\tag{OCl($\rbcalt$).1}\label{lp:center:ordered1}\\
	&&\sum_{j\in\calC}u_j&\geq m\tag{OCl($\rbcalt$).2}\label{lp:center:ordered2}\\
	&&\sum_{i\in\calF,d(i,j)\geq T_\ell}x_{ij}&\leq \ell\quad\;\,\forall j\in\calC,\ell\in\pos\tag{OCl($\rbcalt$).3}\label{lp:center:ordered3}\\
	&&\sum_{\ell\in\pos}\left(\up{\ww}{n}_\ell-\up{\ww}{n}_{\nextp{\ell}}\right)
	\left(\sum_{d(i,j)\geq T_\ell}x_{ij}d(i,j)\right)&\leq B\quad\forall j\in\calC,n\in[N]\tag{OCl($\rbcalt$).4}\label{lp:center:ordered4}\\
	&&\sum_{i\in\calF}y_i&= k.\tag{OCl($\rbcalt$).5}\label{lp:center:ordered5}
	\end{align}

\begin{claim}
	\emph{\ref{lp:center:ordered}} is feasible.
\end{claim}

\begin{proof}
The proof is essentially the same as~\cref{claim:load-ordered}.
\end{proof}

For any feasible solution $(x,u,y)$ of~\ref{lp:center:ordered}, we give it an identical treatment for $\topq{\ell}{q}$ norms using~\cref{algorithm:bundle}, and define $\hat S,\hat S_j$ in the same way.

\begin{lemma}\label{lemma:ordered}
	For any client $j\in\calC$ and any $n\in[N]$, we have\emph{
	\[\sum_{l\in\pos}(\up{\ww}{n}_\ell-\up{\ww}{n}_{\nextp{\ell}})\cdot\topl(\vec{d}(j,\hat S_j))\le 3B+6R\cdot\up{\ww}{n}_1+3\sum_{\ell\in\pos}(\up{\ww}{n}_\ell-\up{\ww}{n}_{\nextp{\ell}})(\ell-1)T_\ell.\]}
\end{lemma}

\begin{proof}
	The proof is essentially the same as~\eqref{eq:center:topl:inline}.
\end{proof}

\begin{lemma}\label{lemma:center-gap}
\begin{equation*}
    \max_{n\in[N]}\left\{\sum_{\ell\in\pos}(\up{\ww}{n}_\ell-\up{\ww}{n}_{\nextp{\ell}})\ell\cdot T_\ell
    \right\}=O(\log|\calJ|)\widetilde{\opt}.
\end{equation*}
\end{lemma}

\begin{proof}
The proof is the same as~\cref{lemma:load-gap}.
\end{proof}

\begin{proof}[Proof of~\cref{theorem:center:main},~\cref{theorem:center:main2}]
By combining~\cref{lemma:ordered} and~\cref{lemma:center-gap}, and recalling that $2\widetilde{\opt}\geq B, \widetilde{\opt}\geq R\cdot\max_{n\in[N]}\{\up{\ww}{n}_1\}$, the approximation ratio is obvious.
\end{proof}

\subsection{MaxOrdered-Norm Matroid Center}\label{app:matcenter:ordered}

As a first extension, we consider \mnmat for {\mon}s. We address the issue between splitting facility locations and the original matroid $\calM=(\calF,\calI)$ with rank function $r$. After the split is over, we define a new matroid $\calM'=(\calF',\calI')$ as follows: \emph{$S\subseteq\calF'$ is independent if and only if $g$ restricted to $S$ is injective and $g(S)$ is independent in $\calM$}. In other words, if and only $g$ maps $S$ to an independent set in $\calM$ and $S$ does not contain two copies of the same facility location in $\calF$. It is not hard to verify that $\calM'$ is a matroid, and we denote its rank function by $r'$.

As we consider problems with matroids, we need the characterization of certain LPs with matroid constraints.
The following technical lemma characterizes the intersection of constraints from a laminar family and a matroid polytope.

\begin{lemma}\label{lemma:laminar-matroid-intersection}
	Let $E$ be a finite ground set, $\calL$ be a laminar family on $E$ and $\calM$ be a matroid on $E$ with rank function $r$. If the polytope $\calP=\{x \in\R^{|E|}: x\geq0, A_1x\leq b_1, A_2x\geq b_2, x(S)\leq r(S)\forall S\subseteq E\}$ satisfies that each row of $A_1,A_2$ corresponds to the characteristic vector $\chi_L$ of some $L\in\calL$ and all entries of $b_1,b_2$ are integers, then $\calP$ either is empty or  has integral extreme points.
\end{lemma}

\begin{proof}
	We provide a proof based on induction of $|E|$. When $|E|$ is 1 the result is trivial. Suppose the theorem holds true for $|E|=1,\dots,n$. Consider the case of $|E|=n+1$ and let $x$ be a vertex of $\calP$,
	\begin{description}
	\item[(1) When $\exists i, x_i=0$.] Delete $i$ from $\calM$ and remove its corresponding column from $A_1,A_2$, also remove the $i$-th entry in $x$ and get $\tilde x$. Obviously $\tilde x$ still satisfies $A_1'\tilde x\leq b_1$ and $A_2'\tilde x\geq b_2$, and $\tilde x$ is still in the matroid polytope of $\calM\backslash i$, therefore from induction, $\tilde x$ is integral.
	\item[(2) When $\exists i, x_i=1$] We contract $\calM$ by $i$, remove its corresponding column from $A_1,A_2$, subtract 1 from corresponding entries in $b_1,b_2$ if the value removed is equal to 1, and remove the $i$-th entry in $x$ to obtain $\tilde x$. It is easy to see that $\tilde x$ still satisfies $A_1'\tilde x\leq b_1'$ and $A_2'\tilde x\geq b_2'$, and for any $S\subseteq E\setminus\{i\}$, we have $\tilde x(S)=x(S+i)-x_i\leq r(S+i)-1=r'(S)$, where $r'$ is the rank function of $\calM/i$. Therefore, using the induction hypothesis, $\tilde x$ is integral.
	\item[(3) For any $i, x_i\in(0,1)$.] Let $\calT_1=\{L\in\calL: x\text{ is tight at constraint of }\chi_L\}$ and $\calT_2=\{S\subseteq E : x(S)=r(S)\}$. It is well-known that there exists a maximal chain $\calC_2=\{C_1,\dots,C_k\}\subseteq \calT_2$ with $\emptyset\subseteq C_1\subseteq\cdots\subseteq C_k$ such that $\spn(\chi_S:S\in\calC_2)=\spn(\calT_2)$.
	
	Now since every constraint of $x_i\geq0$ is not tight, the maximum number of linearly-independent tight constraints corresponding to $\calT_1\cup \calC_2$ must be at least $|E|=n+1$. On the other hand, since $x_i\in(0,1)\forall i$ and every tight constraint is of the form $\langle\chi_S,x\rangle=b$ with $b\in\mathbb{Z}$ and $\chi_S$ being the characteristic vector of $S\subseteq E$, the size of maximum linearly-independent vectors in $\{\chi_L:L\in\calT_1\}$ is at most $(n+1)/2$, since $\calL$ is laminar in the first place. The same argument applies to $\calC_2$, and we see that by combining the arguments above, the sizes of maximum linearly-independent subsets of $\calT_1$ and $\calC_2$ are both exactly $(n+1)/2$, but this puts $E$ in both subsets, and the two subsets cannot be combined into another larger linearly-independent subset of size exactly $n+1$. Therefore $x$ is not a vertex solution, contradiction!\qedhere
	\end{description}
\end{proof}

We use the same algorithm as in~\cref{app:center:ordered}, and only need to replace~\eqref{lp:center:ordered5} with $y(S)\leq r(S)\forall S\subseteq\calF$, and~\eqref{lp:aux:k4} with $z(S)\le r'(S)\forall S\subseteq\calF'$. From the definition of $\calM'$ and its rank function $r'$, we notice the following for $y$ and any $S\subseteq\calF'$,
\[\sum_{i\in S}y_i\leq\sum_{i'\in g(S)}\sum_{i\in g^{-1}(i')}y_i=\sum_{i'\in g(S)}y_{i'}\leq r(g(S)),\]
where we abuse the notation for $i'\in\calF$ and use $y_{i'}$ to represent its value before split. By definition, we have $r'(S)=r(g(S))$, hence the LPs are satisfied by $y$.

On the other hand, the constraint~\eqref{lp:aux:k3} is actually subsumed by the matroid constraints, because for any $i'\in\calF$, we have $r'(g^{-1}(i'))=r(\{i'\})\leq 1$, so the constraints form a laminar family and a matroid. According to~\cref{lemma:laminar-matroid-intersection}, the LP has integral extreme points. We choose one optimal integral solution of the new auxiliary rounding LP as $z^\star$ and let $\hat S=\{g(i):i\in\calF',z^\star_i=1\}$. The following are simple corollaries of~\cref{theorem:center:main}.

\begin{corollary}
	For any $\epsilon>0$, there exists a $(3\cdot4^{1/q}+\epsilon)$-approximation for \mnmat, when the norm $f$ is a \emph{$\topq{\ell}{q}$} norm, and the running time is $poly(n)/\epsilon$.
\end{corollary}

\begin{corollary}
	There exists a polynomial time $O(\log\rma)$-approximation for \mnmat, when the norm $f$ is a \mon with $N$ weight vectors.
\end{corollary}

\subsection{MaxOrdered-Norm Knapsack Center}\label{app:knapcenter:ordered}

In this section, we consider \mnknap for {\mon}s. Using the same process, we replace~\eqref{lp:center:ordered5} with the knapsack constraint $\sum_{i\in\calF}\wt_i\cdot y_i\leq W$, and~\eqref{lp:aux:k4} with $\sum_{i\in\calF'}\wt_i\cdot z_i\leq W$, where all copies of the same facility $i$ inherit the same weight $\wt_i$. The constraints in the auxiliary LP are now the intersection of two laminar families, namely $\calU_1\cup\calU_2$ and  $\{g^{-1}(i'):i'\in\calF\}$, as well as a knapsack constraint. Therefore, we have to drop the constraints corresponding to the second laminar family, and only consider the remaining ones, which is known to produce extreme points that have at most two fractional entries.

We choose one such optimal solution $z^\star$, let $\hat S=\{g(i):i\in\calF',z^\star_i>0\}$ and allow $\hat S$ to choose multiple copies of the same $i'\in\calF$, i.e., $\hat S$ is a multi-set. The proof for the approximation factor of $\hat S$ is the same as the $k$-cardinality case, thus omitted here. Since there are at most two fractional values in $z^\star$, $\hat S$ contains at most two such facilities, therefore the total weight of $\hat S$ is at most $\wt(\hat S)\leq \sum_{z^\star_i=1}\wt_i+2\max_{i\in\calF}\wt_i\leq W+2\max_{i\in\calF}\wt_i$.

    Next, we show how to reduce the violation of the total weight $W$. We use the standard guessing technique and try to guess exactly the set of facilities having weight $\geq\epsilon\cdot W$ in the optimal solution $S^\star$, and define it as $S_0=\{i\in S^\star: \wt_i\geq\epsilon\cdot W\}$. It is obvious that $|S_0|\leq1/\epsilon$, and guessing it takes at most $n^{1/\epsilon}$ possible tries. Assume that we now know $S_0$.
    
    Define $\calF_{<\epsilon}=\{i\in\calF:\wt_i<\epsilon\cdot  W\}$, and we further roughly guess the optimum value $\opt$. More specifically, recall that the largest distance in the optimal solution $R$ is also known to us and $\opt\geq R\cdot\max_{n\in[N]}\{\up{\ww}{n}_1\}$, and on the other hand, we have $\opt\leq \rma R\cdot\max_{n\in[N]}\{\up{\ww}{n}_1\}$, so it takes $O(\log_{1+\epsilon}\rma)$ guesses to find $B\in[\opt,(1+\epsilon)\opt)$.
    
    Next, for any $j\in\calC$ and consider the number of nearby pre-selected facilities, namely $\{i\in S_0:d(i,j)\leq R\}$. From near to far, connect $j$ to as many of them as possible, as long as this connection cost vector has $f$-norm at most $B$. Reduce $l_j,r_j$ and $m$ accordingly.
    
    In the optimal solution, the overall $f$-norm of connections between $j$ and $S_0,\calF_{<\epsilon}$ is obviously at most $\opt\leq B$, therefore by prioritizing the facilities in $S_0$, the new instance also has optimum at most $\opt\leq B$, largest connection at most $R$ and the set of facilities is now $\calF_{<\epsilon}$. Using the previous result, we can obtain an approximation $S_1$ violating the knapsack constraint $W-\wt(S_0)$ by at most $2\max_{i\in\calF_{<\epsilon}}\wt_i\leq2\epsilon\cdot W$, thus
    \[\wt\left(S_0\cup S_1\right)=\wt(S_0)+\wt(S_1)\leq \wt(S_0)+W-\wt(S_0)+2\epsilon\cdot W=(1+2\epsilon)W,\]
    and the connection cost for every $j$ can be further bounded by adding the cost of $S_0$ and $S_1$, using the triangle inequality on $f$.
    
\begin{corollary}
	For any $\epsilon>0$, there exists a $(1+3\cdot4^{1/q}+\epsilon)$-approximation for \mnknap, when the norm $f$ is a \emph{$\topq{\ell}{q}$} norm in running time $n^{O(1/\epsilon)}$. The solution violates the knapsack constraint by a multiplicative factor of $\epsilon$, and allows multiple facilities at the same location.
\end{corollary}

\begin{corollary}
	There exists a $O(\log\rma)$-approximation for \mnknap, when the norm $f$ is a \mon with $N$ weight vectors. The solution violates the knapsack constraint by a multiplicative factor of $\epsilon$, allows multiple facilities at the same location, and runs in time $n^{O(1/\epsilon)}$.
\end{corollary}

\begin{remark*}
Chakrabarty and Negahbani~\cite{chakrabarty2018generalized} give the first pure 3-approximation for robust knapsack center, under a general framework of down-closed family of subset constraints. In our formulation, the requirement of connecting multiple distinct open facilities for a single client is fundamentally more challenging than vanilla center clustering problems. Intuitively speaking, our current framework utilizes two combinatorial structures like laminar families, therefore either the additional knapsack constraint or the coverage constraint (the robustness requirement $m$) will cause trouble for rounding, and the dynamic programming method in~\cite{chakrabarty2018generalized} is also likely to fail. Kumar and Raichel~\cite{kumar2013revisited} consider transforming an approximate non-fault-tolerant solution into a fault-tolerant one, and introduce a constant factor during the transition, but the existence of similar results, particularly for knapsack variants of clustering problems, is yet to be understood.
\end{remark*}

%%%%----%%%%----

\subsection{Proof of~\cref{lemma:fair:center:core}}\label{app:fair:center:core-proof}

First, we define the set of violating subsets for notation simplicity,
\[\calF_c^{\almu}(B)=\left\{S\in\calF_c(B)\left|\sum_{j\in\calC}\alpha_j\cdot \counted_B(j,S) >\mu\right.\right\},\]
and $\calE_c^\almu(B)$ similarly. We easily see that, if $(\almu)$ satisfies $\sum_{j\in\calC}\alpha_j\cdot e_j\geq\mu+1$ and $\calF_c^{\almu}(B)$ is empty, then $(\almu)\in\calQ(B)$. Therefore, to prove~\cref{lemma:fair:center:core}, for any given $(\almu)\in\mathbb{Q}_{\geq0}^{|\calC|}\times\mathbb{Q}$, we need to either certify that $\calF_c^{\almu}(B)$ is empty, or find a subset $S\in\calE_c^{\almu}(3\cdot4^{1/q}B)$. The following lemma encodes the strict inequality as an equivalent non-strict one, and is directly obtained from~\cite{anegg2020technique}, thus we omit the proof here.

\begin{lemma}(\cite{anegg2020technique})
	Let $(\almu)\in\mathbb{Q}_{\geq0}^{|\calC|}\times\mathbb{Q}$ with representation length $L$. Then one can efficiently compute an $\eta>0$ with representation length $O(L)$, such that for any $S\in\calF_c(B)$, we have 
	$\sum_{j\in\calC}\alpha_j\cdot\counted_B(j,S)>\mu$ if and only if $\sum_{j\in\calC}\alpha_j\cdot\counted_B(j,S)\geq\mu+\eta$.
\end{lemma}

We fix such an $\eta>0$ from now on, and further define the modified polytopes
\[\calP_c^\almu(\rbt) = \left\{(x,u,y)\in\text{\ref{lp:fair:center:topl}}\left|\sum_{j\in\calC}\alpha_j\cdot u_j\geq\mu+\eta\right.\right\}.\]

By definition, if $(\almu)$ satisfies $\sum_{j\in\calC}\alpha_j\cdot e_j\geq\mu+1$ but $(\almu)\notin\calQ_c(B)$, then there exists some subset $S\in\calF_c^\almu(B)$, hence $S$ induces an integral solution to $\calP_c^\almu(\rbt)$ for \emph{some} $(\rbt)$, by letting $u_j=\counted_B(j,S)$ and $\calP_c^\almu(\rbt)$ is thus non-empty. The contrapositive of this observation tells us the following.
\begin{observation}\label{observation:fair:center}
Fix $B\geq0$. If $(\almu)$ satisfies $\sum_{j\in\calC}\alpha_j\cdot e_j\geq\mu+1$ and $\calP_c^\almu(\rbt)$ is empty \emph{for any} $(\rbt)$ that satisfies~\eqref{eq:fair:center:guess}, then $(\almu)\in\calQ_c(B)$.
\end{observation}

\begin{lemma}\label{lemma:center:separating}
	Let $(\almu)\in\mathbb{Q}_{\geq0}^{|\calC|}\times\mathbb{Q}$ and $(\rbt)$ satisfies~\eqref{eq:fair:center:guess}. There is a polynomial-time algorithm $\calA_0$ that, either returns $S\in\calE_c^\almu(3\cdot4^{1/q}B)$ or certifies $\calP_c^\almu(\rbt)$ is empty.
\end{lemma}

\begin{proof} $\calA_0$ simply returns if $\calP_c^\almu(\rbt)$ is empty. Assume otherwise and we have some $(x,u,y)\in\calP_c^\almu(\rbt)$, where $(\rbt)$ satisfies~\eqref{eq:fair:center:guess}. We run $\algb$ and obtain the output bundles. Specifically, $\calU_1$ is the set of full bundles, and $\calU_2$ is the set of partial bundles. All bundles are pair-wise disjoint, though may contain different copies of the same original facility location.

In a similar fashion to~\cref{lemma:aux:k}, we define $\beta_U=\sum_{j\in \calC:U\in\queue_j}\alpha_{j}$ \emph{only for} $U\in\calU_2$, and claim the inequality
\[\sum_{j\in\calU_2}\beta_U\cdot y(U)+\sum_{j\in\calC}\sum_{U\in\queue_j}\alpha_j\cdot\mathbbm{1}[U\in\calU_1]\geq \mu+\eta.\]
Indeed, the inequality above can be easily verified by considering $(x,u,y)$ satisfies $\sum_{j\in\calC}\alpha_j\cdot u_j\geq\mu+\eta$, and the fact that the contribution on the LHS is at least $\sum_{j}\alpha_j\cdot u_j$.

Now we consider the following auxiliary LP derived from $\calP_c^\almu(\rbt)$ and $(x,u,y)$,
	\begin{alignat*}{2}
		\text{max\quad}&& \sum_{j\in\calU_2}\beta_U\cdot z(U)&+\sum_{j\in\calC}\sum_{U\in\queue_j}\alpha_j\cdot\mathbbm{1}[U\in\calU_1]\tag{$\mathrm{ACl}^\almu(\rbt)$}\label{lp:aux:fair}\\
		\text{s.t.\quad}&& z(U)&=1\quad\quad\forall U\in\calU_1\\
		&& z(U)&\leq1\quad\quad\forall U\in\calU_2\\
		&& z(g^{-1}(i'))&\leq1\quad\quad\forall i'\in\calF\\
		&& z(\calF') &\leq k\\
		&& z_i&\geq0\quad\quad\forall i\in\calF',
	\end{alignat*}
where we recall that $\calF'$ is the set of extended facility locations, after making copies, and $g^{-1}:\calF\rightarrow 2^{\calF'}$ takes the original copy $i'$ to the set of all its copies. Since $y$ satisfies all the constraints, and the constraints form two laminar families on $\calF'$, the optimal solution $z^\star$ is integral with the objective value associated with $z^\star$ being at least $\mu+\eta$.
	
Finally, we define $S=\{g(i):z^\star_i=1\}$, $u^\star_j= z^\star(\queue_j)$ and the assignment variables $x^\star_{ij}$ accordingly. It is not hard to check that $(x^\star,u^\star,z^\star)$ is an integral solution corresponding to $S\in\calE_c^\almu(3\cdot4^{1/q}B)$, using a proof identical to~\cref{theorem:center:main}.
\end{proof}

\begin{proof}[Proof of~\cref{lemma:fair:center:core}] We use the algorithm $\calA_0$ in~\cref{lemma:center:separating} as part of the separation hyperplane oracle, in order to check whether $\calP_c^\almu(\rbt)$ is empty for any feasible $(\rbt)$ . Fix $B,(\almu)$ and iterate through all possible $(\rbt)$ that satisfies~\eqref{eq:fair:center:guess}, and start with any $(x,u,y)$ in each iteration: 
\begin{itemize}
    \item Whenever the algorithm $\calA_0$ is called and returns some $S\in\calE_c^\almu(3\cdot4^{1/q}B)$, then $(\almu)$ violates $\calQ_c(3\cdot4^{1/q}B)$ and we simply return $S$;
    \item Assume $\calA_0$ verifies the emptiness of $\calP_c^\almu(\rbt)$ for all $(\rbt)$, in polynomial time. Using~\cref{observation:fair:center}, one has $(\almu)\in\calQ_c(B)$ verified.\qedhere
\end{itemize} 
\end{proof}

\end{document}